\documentclass[aos,preprint]{imsart}
\usepackage[]{graphicx}\usepackage[]{color}

\makeatletter
\def\maxwidth{ %
  \ifdim\Gin@nat@width>\linewidth
    \linewidth
  \else
    \Gin@nat@width
  \fi
}
\makeatother

\definecolor{fgcolor}{rgb}{0.345, 0.345, 0.345}

\usepackage{framed}
\makeatletter
 {\par\unskip\endMakeFramed%
 \at@end@of@kframe}
\makeatother

\definecolor{shadecolor}{rgb}{.97, .97, .97}
\definecolor{messagecolor}{rgb}{0, 0, 0}
\definecolor{warningcolor}{rgb}{1, 0, 1}
\definecolor{errorcolor}{rgb}{1, 0, 0}

\usepackage{alltt} 
\usepackage[latin1]{inputenc}
\usepackage{enumitem}
\usepackage{todonotes}
\usepackage{latexsym}
\usepackage{amssymb}
\usepackage{epsfig}
\usepackage{amsmath}
\usepackage{xcolor}
\usepackage{float}
\usepackage{hyperref}
\usepackage{subfig}
\usepackage{multirow}
\usepackage{booktabs}

\IfFileExists{upquote.sty}{\usepackage{upquote}}{}

\usepackage[latin1]{inputenc}
\usepackage{enumitem}
\usepackage{todonotes}
\RequirePackage[OT1]{fontenc}
\RequirePackage{amsmath,amssymb,amsthm}
\RequirePackage[numbers]{natbib}
\RequirePackage[colorlinks,citecolor=blue,urlcolor=blue]{hyperref}
\usepackage{graphicx}
\usepackage{mathtools}
\DeclarePairedDelimiter{\ceil}{\lceil}{\rceil}

\startlocaldefs
\numberwithin{equation}{section}
\theoremstyle{plain}

\endlocaldefs

\theoremstyle{plain}
\long\def\comment#1{}
\newtheorem{algorithm}{Algorithm}
\newtheorem{theorem}{Theorem}

\newtheorem{lemma}{Lemma}

\theoremstyle{definition}

\numberwithin{definition}{section}

\newtheorem{remark}{Comment}[section]
\numberwithin{remark}{section}

\DeclareMathOperator{\diag}{diag}

\newcommand{\R}{\mathbb{R}}

\newcommand{\N}{\mathbb{N}}

\renewcommand{\qed}{\hfill{\tiny \ensuremath{\blacksquare} }}%

\newcommand{\bP}{\mathbb{P}}
\newcommand{\E}{\mathbb{E}}

\makeatletter
\newcommand{\setword}[2]{%
  \phantomsection
  #1\def\@currentlabel{\unexpanded{#1}}\label{#2}%
}
\makeatother

\begin{document}

\begin{frontmatter}
\title{Uniform Inference in High-Dimensional Gaussian Graphical Models\thanksref{T1}}
\runtitle{Inference in Gaussian Graphical Models}
\thankstext{T1}{Version November 2018.}

\begin{aug}
\author{\fnms{Sven} \snm{Klaassen}\ead[label=e1]{}},
\author{\fnms{Jannis} \snm{K\"uck}\ead[label=e2]{}},\\
\author{\fnms{Martin} \snm{Spindler}\ead[label=e3]{}}
\and
\author{\fnms{Victor} \snm{Chernozhukov}\ead[label=e4]{}}


\address{Sven Klaassen\\
University of Hamburg\\
Hamburg Business School\\
Moorweidenstr. 18\\
20148 Hamburg\\
Germany\\
E-mail: sven.klaassen@uni-hamburg.de}

\address{Jannis K\"uck\\
University of Hamburg\\
Hamburg Business School\\
Moorweidenstr. 18\\
20148 Hamburg\\
Germany\\
E-mail: jannis.kueck@uni-hamburg.de}

\address{Martin Spindler\\
University of Hamburg\\
Hamburg Business School\\
Moorweidenstr. 18\\
20148 Hamburg\\
Germany\\
E-mail: martin.spindler@uni-hamburg.de}

\address{Victor Chernozhukov\\
Massachusetts Institute of Technology\\
Department of Economics and\\ Operations Research Center\\
50 Memorial Drive\\
Cambridge, MA 02142\\
USA\\
E-mail: vchern@mit.edu}
\end{aug}

\begin{abstract}
Graphical models have become a very popular tool for representing dependencies within a large set of variables and are key for representing causal structures.  
We provide results for uniform inference on high-dimensional graphical models with the number of target parameters $d$ being possible much larger than sample size. This is in particular important when certain features or structures of a causal model should be recovered. Our results highlight how in high-dimensional settings graphical models can be estimated and recovered with modern machine learning methods in complex data sets. To construct simultaneous confidence regions on many target parameters, sufficiently fast estimation rates of the nuisance functions are crucial. In this context, we establish uniform estimation rates and sparsity guarantees of the square-root estimator in a random design under approximate sparsity conditions that might be of independent interest for related problems in high-dimensions. We also demonstrate in a comprehensive simulation study that our procedure has good small sample properties.
\end{abstract}


\end{frontmatter}
\section{Introduction}
We provide methodology and theory for uniform inference on high-dimensional graphical models with the number of target parameters being possible much larger than sample size. We demonstrate uniform asymptotic normality of the proposed estimator over $d$-dimensional rectangles and construct simultaneous confidence bands on all of the $d$ target parameters. The proposed method can be applied to test simultaneously the presence of a large set of edges in the graphical model
$$X=(X_{1},\dots,X_{p})^T\sim\mathcal{N}(\mu_X,\Sigma_X).$$
Assuming that the covariance matrix $\Sigma_X$ is nonsingular, the conditional independence structure
of the distribution can be conveniently represented by a graph $G = (V,E)$, where $V =\{1, \dots , p\}$ is the set of nodes and $E$ the set of edges
in $V \times V$. Every pair of variables not contained in the edge set is conditionally independent given all remaining variables. If the vector $X$ is normally distributed, every edge corresponds to a non-zero entry in the inverse covariance matrix (Lauritzen (1996)) \cite{lauritzen1996graphical}.\\ \\
In the last decade, significant progress has been made on estimation of a large precision matrix in order to analyze the dependence structure of a high-dimensional normal distributed random variable. There are mainly two common approaches to estimate the entries of a precision matrix. The first approach is a penalized likelihood estimation approach with a lasso-type penalty on entries of the precision matrix, typically referred to as the graphical lasso. This approach has been studied in several papers, see e.g Lam and Fan (2009) \cite{lam2009sparsistency}, Rothman et al. (2008) \cite{rothman2008sparse}, Ravikumar et al. (2011) \cite{ravikumar2011high} and  Yuan and Lin (2007) \cite{yuan2007model}.
The second approach, first introduced by Meinshausen and B\"uhlmann (2006) \cite{meinshausen2006high}, is neighborhood based. It estimates the conditional independence restrictions separately for each node in the graph and is hence equivalent to variable selection for Gaussian linear models. The idea of estimating the precision matrix column by column by running a regression for each variable against the rest of variables was further studied in Yuan (2010) \cite{yuan2010high}, Cai, Liu and Zhou (2011) \cite{cai2011constrained} and Sun and Zhang (2013) \cite{sun2013sparse}.
\\ \\
In this paper, we do not aim to estimate the whole precision matrix but we focus on quantifying the uncertainty of recovering its support by providing a significance test for a set of potential edges. In recent years, statistical inference for the precision matrix in high-dimensional settings has been studied, e.g in Jankov\'{a} and van de Geer (2016) \cite{jankova2017honest} and Ren et al. (2015) \cite{ren2015asymptotic}. Both approaches lead to an estimate that is elementwise asymptotically normal and enables testing for low-dimensional parameters of the precision matrix using standard procedures such as Bonferroni-Holm correction.\\
In contrast to these existing results, our method explicitly allows for testing a joint hypothesis without correction for multiple testing and conducting inference for a growing number of parameters using high dimensional central limit results. In particular, our results rely on approximate sparsity instead of row sparsity which restricts the number of non-zero entries of each row of the precision matrix to be at most $s\ll n$ that is in many applications a questionable assumption.  In order to provide theoretical results, fitting the problem of support discovery in Gaussian graphical models into a general Z-estimation setting with a high-dimensional nuisance function is key. Inference on a (multivariate) target parameter in general Z-estimation problems in high dimensions is covered in Belloni et al. (2014) \cite{belloni2014uniform}, Belloni et al. (2018) \cite{belloni2018uniformly} and Chernozhukov et al. (2017) \cite{chernozhukov2017double}. To conduct inference on a high-dimensional target parameter, uniform estimation rates and sparsity guarantees of the nuisance function are crucial. In this context, we formally apply recent results from Belloni et al. (2018) \cite{belloni2018uniformly} to ensure sufficient fast convergence rate of the lasso estimator under approximate sparsity conditions. Moreover, we provide auxiliary results for the square-lasso estimator establishing uniform estimation rates and sparsity guarantees in a random design under approximate sparsity conditions that might be of independent interest for related problems in high-dimensional linear models.
\section*{Plan of this Paper}
The rest of this paper is organized as follows. In Section \ref{Setting}, we formally define the setting and introduce the notation that will be used fitting
the problem of support discovery in Gaussian graphical models into a general Z-estimation problem with a high-dimensional nuisance function. In Section \ref{estimation}, we outline the estimation procedure of the high-dimensional target parameter and the conditions that are needed to achieve our main theorem presented in Section \ref{mainsec}. Section \ref{implementation} provides implementation details and shows how our estimation procedure can be modified by cross-fitting to improve small sample properties. Section \ref{simulation} provides a simulation study on the proposed method. The supplementary material includes additional technical material. The proof of our main theorem is provided in Appendix \ref{appendixmain}. The uniform nuisance function estimation is discussed in Appendix \ref{uniformestimation}. Appendix \ref{uniformla} formally discusses conditions for the uniform convergence rates of the lasso estimator. Finally, Appendix \ref{uniformsq} provides auxiliary results for the square-lasso estimator.
\section{Setting} \label{Setting}
Let 
$$X=(X_{1},\dots,X_{p})^T\sim\mathcal{N}(\mu_X,\Sigma_X)$$
be a $p$-dimensional random variable. For all $(j,k)\in E$ with $j\neq k$, assume that
$$X_j=\sum\limits_{\substack{l =1\\ l\neq j}}^p \beta_{l}^{(j)}X_l+\varepsilon^{(j)}=\beta^{(j)} X_{-j}+\varepsilon^{(j)}$$
and
$$X_k=\gamma^{(j,k)}X_{-\{j,k\}}+\nu^{(j,k)},$$
where $\E [\varepsilon^{(j)}|X_{-j}]=0$ and $\E [X_{-\{j,k\}}\nu^{(j,k)}]=0$. Define the column vector
$$\Gamma^{(j)}=\left(-\beta^{(j)}_1,\dots,-\beta^{(j)}_{j-1},1,-\beta^{(j)}_{j+1},\dots,-\beta^{(j)}_{p}\right)^T.$$
One may show
\begin{align*}
\Phi_0=\left(\Phi_0^1,\dots,\Phi_0^p\right)=\left(\Gamma^{(1)}/Var(\varepsilon^{(1)}),\dots,\Gamma^{(p)}/Var(\varepsilon^{(p)})\right),
\end{align*}
where $\Phi_0^{j}$ is the $j$-th column of the precision matrix $\Phi_0=\Sigma_X^{-1}$ \cite{jankova2017honest}.
Hence
\begin{align}\label{beta0}
\beta^{(j)}_{k}=0\Leftrightarrow \beta^{(k)}_{j}=0\Leftrightarrow X_j\perp X_k|X_{-\{j,k\}}
\end{align}
for all $j\neq k$. Assume that we are interested in the following set of potential edges
$$\mathcal{M}:=\{m_1,\dots,m_{d_n}\}$$
where the number of edges $d_n$ may increase with sample size $n$. In the following the dependence on $n$ is omitted to simplify the notation.
In order to test whether all variables $X_j$ and $X_k$ are conditionally independent with $m_r=(j_r,k_r)$ for a $r\in\{1,\dots,d\}$, we have to estimate our target parameter
$$\theta_0=(\theta_{m_{1}},\dots,\theta_{m_{d}})^T:=(\beta^{(j_1)}_{k_1},\dots,\beta^{(j_d)}_{k_d})^T.$$
The setting above fits in the general Z-estimation problem of the form
$$\E \left[\psi_{m_r}\big(X,\theta_{m_r},\eta_{m_r}\big)\right]=0$$
for all $r=1,\dots,d$ with nuisance parameters $$\eta_{m_r}=\left(\beta^{(j)}_{-k},\gamma^{(j,k)}\right)$$ where $\beta^{(j)}_{-k}\equiv\beta^{(m_r)}$ and $\gamma^{(j,k)}\equiv\gamma^{(m_r)}$. The score functions are defined by 
\begin{align*}
\psi_{m_r}(X,\theta,\eta):&=\Big(X_j-\theta X_k-\eta^{(1)}X_{-m_r}\Big)\Big(X_k-\eta^{(2)}X_{-m_r}\Big)
\end{align*}
for $m_r=(j_r,k_r)\equiv (j,k)$, $\eta=(\eta^{(1)},\eta^{(2)})$ and $r=1,\dots,d$. Without loss of generality we assume $j>k$ for all tuples $m_r\in \mathcal{M}$.
\begin{remark}\label{linear}
The score function $\psi$ is linear, meaning 
\begin{align*}
\psi_{m_r}(X,\theta,\eta)=\psi_{m_r}^{a}(X,\eta^{(2)})\theta+\psi_{m_r}^b(X,\eta)
\end{align*}
with
$$\psi^{a}_{m_r}(X,\eta^{(2)})=-X_k\Big(X_k-\eta^{(2)}X_{-m_r}\Big)$$
and
$$\psi^{b}_{m_r}(X,\eta)=\Big(X_j-\eta^{(1)}X_{-m_r}\Big)\Big(X_k-\eta^{(2)}X_{-m_r}\Big)$$
for $m_r=(j,k)$ and $r=1,\dots,d$.\\ \\
It is well known that in partially linear regression models $\theta_0$ satisfies the moment condition
\begin{align}\label{moment condition}
\E \left[\psi_{m_r}\big(X,\theta_{m_r},\eta_{m_r}\big)\right]=0
\end{align}
for all $r=1,\dots,d$ and also the \textit{Neyman orthogonality} condition
\begin{align*}
\partial_{t}\left\{\E\left[\psi_{m_r}\big(X,\theta_{m_r},\eta_{m_r}+t\tilde{\eta}\big)\right]\right\}\big|_{t=0}
\end{align*}
for all $\tilde{\eta}$ in an appropriate set where $\partial_{t}$ denotes the derivate with respect to $t$. These properties are crucial for valid inference in high dimensional settings. We will show these properties explicitly in the proof of Theorem \ref{maintheo}.
\end{remark}
\section{Estimation}\label{estimation}
Let $X^{(i)}$, $i=1,\dots,n$ be i.i.d. random vectors.\\
At first we estimate the nuisance parameter $\eta_{m_r}=\big(\eta_{m_r}^{(1)},\eta_{m_r}^{(2)}\big)$ by
running a lasso/ post-lasso/ square-root lasso regression of $X_j$ on $X_{-j}$ to compute $(\tilde{\theta}_{m_r},\hat{\eta}_{m_r}^{(1)})$ and a lasso/ post-lasso/ square-root lasso regression of $X_k$ on $X_{-m_r}$ to compute $\hat{\eta}_{m_r}^{(2)}$ for each $(j,k)=m_r\in\mathcal{M}$. The estimator $\hat{\theta}_0$ of the target parameter
$$\theta_0=(\theta_{m_1},\dots,\theta_{m_{d}})^T$$
is defined as the solution of
\begin{align}\label{estimator}
\sup\limits_{r=1,\dots,d}\left\{\left|\mathbb{E}_n^{}\Big[\psi_{m_r}\big(X,\hat{\theta}_{m_r},\hat{\eta}_{m_r}\big)\Big]\right|-\inf_{\theta\in\Theta_{m_r}}\left|\mathbb{E}_n^{}\Big[\psi_{m_r}\big(X,\theta,\hat{\eta}_{m_r}\big)\Big]\right|\right\}\le\epsilon_{n},
\end{align}
where $\epsilon_{n}=o\left(\delta_nn^{-1/2}\right)$ is the numerical tolerance and $(\delta_n)_{n\ge 1}$ a sequence of positive constants converging to zero.\\ \\
Assumptions \textbf{A1}-\textbf{A4}.\\
Let $a_n:=\max(d,p,n,e)$ and $C$ a strictly positive constant independent of $n$ and $r$.
The following assumptions hold uniformly in $n\ge n_0,P\in\mathcal{P}_n$:
\begin{enumerate}[label=\textbf{A\arabic*},ref=A\arabic*]
\item\label{A1}
\begin{em}
For all $m_r=(j,k)\in \mathcal{M}$ with $j\neq k$ we have the following approximate sparse representations 
\begin{itemize}
\item[(i)] It holds
\begin{align*}
X_j&=\beta^{(j)} X_{-j}+\varepsilon^{(j)}\\
&=\theta_{m_r} X_{k}+\left(\beta^{(1,m_r)}+\beta^{(2,m_r)}\right)X_{-m_r}+\varepsilon^{(m_r)}
\end{align*}
with $$\|\beta^{(1,m_r)}\|_0\le s,\quad\max_{r=1,\dots,d}\|\beta^{(2,m_r)}\|_1^2\le C\sqrt{\frac{s^2\log(a_n)}{n}}$$
and $$\max_{r=1,\dots,d}\mathbb{E}\left[\left(\beta^{(2,m_r)}X_{-m_r}\right)^2\right]\le C\frac{s\log(a_n)}{n}.$$
\item[(ii)] It holds
\begin{align*}
X_k&=\gamma^{(j,k)}X_{-\{j,k\}}+\nu^{(j,k)}\\
&=\left(\gamma^{(1,m_r)}+\gamma^{(1,m_r)}\right)X_{-m_r}+\nu^{(m_r)}
\end{align*}
with $$\|\gamma^{(1,m_r)}\|_0\le s,\quad\max_{r=1,\dots,d}\|\gamma^{(2,m_r)}\|_1^2\le C\sqrt{\frac{s^2\log(a_n)}{n}}$$
and $$\max_{r=1,\dots,d}\mathbb{E}\left[\left(\gamma^{(2,m_r)}X_{-m_r}\right)^2\right]\le C\frac{s\log(a_n)}{n}.$$
\end{itemize}
\end{em}
\item\label{A2}
\begin{em}
There exist positive numbers $\tilde{q}>0$ and $\kappa<1$ such that the following growth conditions are fulfilled:
\begin{align*}
n^{\frac{1}{\tilde{q}}}\frac{s^2\log^4(a_n)}{n}=o(1)\text{,}\quad\log(d)=o\left(n^{\frac{1}{9}}\wedge n^{\frac{\kappa}{\tilde{q}}}\right).
\end{align*}
\end{em}
\item\label{A3}
\begin{em} 
For all $m_r=(j,k)\in \mathcal{M}$ it holds
$$\|\beta^{(m_r)}\|_2 + \|\gamma^{(m_r)}\|_2 \le C$$
and $$\sup\limits_{r=1,\dots,d}\sup\limits_{\theta\in\Theta_{m_r}}|\theta|\le C.$$
Additionally $\Theta_{m_r}$ contains a ball of radius $\log(\log(n))n^{-1/2}\log^{1/2}(d)\log(n)$ centered at $\theta_{m_r}$.
\end{em}
\item\label{A4}
\begin{em}
It holds
\begin{align*}
\inf\limits_{\|\xi\|_2=1} \E\left[(\xi X)^2\right]\ge c \text{ and } \sup\limits_{\|\xi\|_2=1} \E\left[(\xi X)^2\right]\le C.
\end{align*}
\end{em}
\end{enumerate}
The condition \ref{A1} is a standard approximate sparsity condition that is discussed in more detail in comment \ref{remarksparsity}. The number of relevant  variables $s_n\equiv s$ captured by the regression coefficient $\beta^{(1,m_r)}$ respectively $\gamma^{(1,m_r)}$ can grow with the sample size. The coefficient $\beta^{(2,m_r)}$ respectively $\gamma^{(2,m_r)}$ is the approximate sparse part of the true regression coefficient. This misspecification of a sparse model is controlled by condition \ref{A1}. The growth condition \ref{A2} ensures that $s^2\log^4(a_n)/n$ converges towards zero with at least polynomial speed. If this convergence is too slow ($\tilde{q}\ge 9$) the condition on the growth rate of the number of tested edges become more restrictive. In general, both the number of parameters $p$ and the number of relevant variables $s$ can grow with the sample size in a balanced way. If $s$ is fixed, the number of potential parameters
$p$ can grow at an exponential rate with the sample size. This means that the set of potential variables can be much larger than the sample size, only the number of relevant variables $s$ has to be smaller than the sample size. This situation is common for Lasso-based estimators. Condition \ref{A3} restricts the parameter spaces and ensures that the true coefficients are well behaved. The condition \ref{A4} is a standard eigenvalue condition that restricts the correlation between the components of $X$ and bounds the variances of each $X_j$ from below and above. Assumptions \ref{A1}-\ref{A4} combined with the normal distribution of $X$ imply the conditions \ref{tails}-\ref{growthc} from theorem \ref{uniformlasso} which enables us to estimate the nuisance parameter sufficiently fast by lasso and post-lasso. To ensure a sufficiently fast convergence rate and sparsity guarantees of the square-root lasso estimator further model assumptions are needed. 
\begin{remark}\label{remarksparsity}
If we have exact sparsity for each $\beta^{(k)}$ with $(j,k)\in\mathcal{M}_r$ the sparsity of $\gamma^{(m_r)}$ follows directly.\\
Observe that for $k\in\{1,\dots,p\}\setminus \{j\}$ and $l\in\{1,\dots,p\}\setminus \{j,k\}$ we have
$$\beta^{(k)}_l=0 \Leftrightarrow X_k\perp X_l|X_{-\{k,l\}}\Leftrightarrow \E[X_k X_l| X_{-\{k,l\}}]=0$$
which implies
$$\E[X_k X_l|X_{-\{j,k,l\}}]=\E\left[\E[X_k X_l| X_{-\{k,l\}}]|X_{-\{j,k,l\}}\right]=0$$
and thereby
$$\gamma_l^{(j,k)}=0 \Leftrightarrow X_k\perp X_l|X_{-\{j,k,l\}}\Leftrightarrow \E[X_k X_l| X_{-\{j,k,l\}}]=0.$$
Hence, the sparsity conditions for testing on an edge $(j,k)$ are satisfied if each node $j$ and $k$ is only sparsely connected to all other nodes.
\end{remark} 


\section{Main results}\label{mainsec}
We will prove that the assumptions of Corollary $2.2$ from Belloni et al. (2018) \cite{belloni2018uniformly} hold and hence we are able to use their results to construct confidence intervals even for a growing number of hypothesis $d=d_n$. Define
\begin{align*}
J_{m_r}&:=\partial_\theta\E[\psi_{m_r}(X),\theta,\eta_{m_r}]\big|_{\theta=\theta_{m_r}}=-\E[X_k(X_k-\eta_{m_r}^{(2)}X_{-m_r})]\\
\sigma_{m_r}^2&:=\mathbb{E}\left[J_{m_r}^{-2}\psi_{m_r}^2(X,\theta_{m_r},\eta_{m_r})\right]
\end{align*}
and the corresponding estimators 
\begin{align*}
\hat{J}_{m_r} &=-\E_n[X_k(X_k-\hat{\eta}_{m_r}^{(2)}X_{-m_r})]\\
\hat{\sigma}_{m_r}^2&=\mathbb{E}_n\left[\hat{J}_{m_r}^{-2}\psi_{m_r}^2(X,\hat{\theta}_{m_r},\hat{\eta}_{m_r})\right]
\end{align*}
for $r=1,\dots,d$. To construct confidence intervals we will employ the Gaussian multiplier bootstrap. Define
$$\hat{\psi}_{m_r}(X):=-\hat{\sigma}_{m_r}^{-1}\hat{J}_{m_r}^{-1}\psi_{m_r}(X,\hat{\theta}_{m_r},\hat{\eta}_{m_r})$$
and the process
$$\hat{\mathcal{N}}:=\left(\hat{\mathcal{N}}_{m_r}\right)_{m_r\in\mathcal{M}}=\left(\frac{1}{\sqrt{n}}\sum\limits_{i=1}^n\xi_i\hat{\psi}_{m_r}\big(X^{(i)}\big)\right)_{m_r\in\mathcal{M}}$$
where $(\xi_i)_{i=1}^n$ are independent standard normal random variables which are independent from $\big(X^{(i)}\big)_{i=1}^n$.
We define $c_{\alpha}$ as the $(1-\alpha)$-conditional quantile of $\sup_{m_r\in\mathcal{M}}|\hat{\mathcal{N}}_{m_r}|$ given the observations $\big(X^{(i)}\big)_{i=1}^n$. The following theorem is the main result of our paper and establishes simultaneous confidence bands for the target parameter $\theta_0$.


\begin{theorem}\label{maintheo}\ \\
Under the assumptions \ref{A1}-\ref{A4} with probability $1-o(1)$ uniformly in $P\in \mathcal{P}_n$ the estimator $\hat{\theta}$ in (\ref{estimator}) obeys
\begin{align}
P\left(\hat{\theta}_{m_r}-\frac{c_\alpha\hat{\sigma}_{m_r}}{\sqrt{n}}\le \theta_{m_r}\le \hat{\theta}_{m_r}+\frac{c_\alpha\hat{\sigma}_{m_r}}{\sqrt{n}}, r=1,\dots,d \right)\to 1-\alpha.
\end{align}
\end{theorem}\ \\
Using theorem \ref{maintheo} we are able to construct standard confidence regions which are uniformly valid over a large set of variables and we can check null hypothesis of the form:
$$H_0: \mathcal{M}\cap E = \emptyset.$$
\begin{remark}\label{criticalregion}
Theorem \ref{maintheo} is basically an application of the gaussian approximation and multiplier bootstrap for maxima of sums of high-dimensional random vectors \cite{chernozhukov2013gaussian}. The central limit theorem and bootstrap in high dimension introduced by Chernozhukov, Chetverikov, Kato et al. (2017) \cite{chernozhukov2017central} extend these results to more general sets, more precisely sparsely convex sets. Hence our main theorem can be easily generalized to various confidence regions that contain the true target parameter with probability $1-\alpha$. Theorem \ref{maintheo} provides critical regions of the form
\begin{align}
\sup\limits_{r=1,\dots,d}\left|\sqrt{n}\frac{\hat{\theta}_{m_r}}{\hat{\sigma}_{m_r}}\right|>c_{1-\alpha}.\label{cube}
\end{align}
Alternatively, we can reject the null hypothesis if
\begin{align}
\sup\limits_{r=1,\dots,d}\left|\sqrt{n}\frac{\hat{\theta}_{m_r}}{\hat{\sigma}_{m_r}}\right|<c_{\frac{\alpha}{2}} \quad\text{or}\quad\sup\limits_{r=1,\dots,d}\left|\sqrt{n}\frac{\hat{\theta}_{m_r}}{\hat{\sigma}_{m_r}}\right|>c_{1-\frac{\alpha}{2}}.\label{sphere}
\end{align}
Both of these regions are based on the central limit theorem for hyperrectangles in high dimensions.
The confidence region (\ref{sphere}) is motivated by the fact that the standard normal distribution $\mathcal{N}(0,I_d)$ in high dimensions is concentrated in a thin spherical shell around the sphere of radius $\sqrt{d}$ as described by Roman Vershynin (2017) \cite{vershynin2017high} and therefore might have smaller volume.
More generally, define
\begin{align*}
\hat{\theta}^*_{m_r}(S,exp)=\sum\limits_{s=1}^S\left(\sqrt{n}\frac{\hat{\theta}_{m_{r-s}}}{\hat{\sigma}_{m_{r-s}}}\right)^{exp}
\end{align*}
for a fix $S$, $exp\in\{1,2\}$ and 
\begin{align*}
r-s:=\begin{cases}r-s\ &\text{if}\quad r-s>0 \\ d+(r-s)\ &\text{otherwise}\end{cases}.
\end{align*}
A test that reject the null hypothesis if
\begin{align}
\sup\limits_{r=1,\dots,d}\left|\hat{\theta}^*_{m_r}(S,exp)\right|>c^*_{1-\alpha}\label{S-sparse}
\end{align}
has level $\alpha$ by \cite{chernozhukov2017central}, since the constructed confidence regions correspond to S-sparsely convex sets. Here, $c^*_{1-\alpha}$ is the $(1-\alpha)$-conditional quantile of $\sup_{m_r\in\mathcal{M}}|\hat{\mathcal{N}}^*_{m_r}|$ given the observations $\big(X^{(i)}\big)_{i=1}^n$ with 
$$\hat{\mathcal{N}}^*_{m_r}=\sum\limits_{s=1}^S\left(\hat{\mathcal{N}}_{m_{r-s}}\right)^{exp}$$
where
\begin{align*}
r-s:=\begin{cases}r-s\ &\text{if}\quad r-s>0 \\ d+(r-s)\ &\text{otherwise.}\end{cases}
\end{align*}
\end{remark}


\section{Notes on the implementation} \label{implementation}
We implemented a function that will be added to the $R$-package $hdm$ and estimates the target coefficients
$$(\theta_{m_{1}},\dots,\theta_{m_{d}})^T=(\beta^{(j_1)}_{k_1},\dots,\beta^{(j_d)}_{k_d})^T$$
corresponding the considered set of potential edges
$$\mathcal{M}:=\{m_1,\dots,m_{d_n}\}$$
by the proposed method described in section \ref{estimation}. It can be used to perform hypothesis tests with asymptotic level $\alpha$ based on the different confidence regions described in comment \ref{criticalregion}. The nuisance function can be estimated by lasso, post-lasso or square-root lasso. 
\subsection{Cross-fitting} In general Z- estimation problems where a so called debiased or double machine learning (DML) method is used to construct confidence intervals, it is common to use cross-fitting in order to improve small sample properties. A detailed discussion of cross-fitted DML can be found in Chernozhukov et al. (2017) \cite{chernozhukov2017double}. The following algorithm generalizes our proposed method to a $K$-fold cross fitted version. We assume that $n$ is divisible by $K$ in order to simplify notation.  
\begin{algorithm}\label{DMLk}
1) Take a $K$-fold random partition $(I_k)_{k=1}^K$ of observation indices $[n]=\{1,\dots,n\}$ such that the size of each fold $I_k$ is $N$. Also, for each $k\in[K]=\{1,\dots,K\}$, define $I_k^c:=\{1,\dots,N\}\setminus I_k$. 2) For each $k\in [K]$ and $r=1,\dots,d$, construct an estimator 
$$\hat{\eta}_{k,m_r}=\hat{\eta}_{m_r}\left((X_i)_{i\in I_k^c}\right)$$
by lasso/ post-lasso or square-root lasso. 3) For each $k\in [K]$, construct an estimator $\hat{\theta}_k=(\hat{\theta}_{k,m_1},\dots,\hat{\theta}_{k,m_d})$ as in \ref{estimator}:
\begin{align*}
&\quad\sup\limits_{r=1,\dots,d}\left\{\left|\mathbb{E}_{N,k}^{}\Big[\psi_{m_r}\big(X,\hat{\theta}_{k,m_r},\hat{\eta}_{k,m_r}\big)\Big]\right|-\inf_{\theta\in\Theta_{m_r}}\left|\mathbb{E}_{N,k}^{}\Big[\psi_{m_r}\big(X,\theta,\hat{\eta}_{k,m_r}\big)\Big]\right|\right\}\\
&\le\epsilon_{n}
\end{align*}
with $\mathbb{E}_{N,k}[\psi_{m_r}(X_i)]=N^{-1}\sum_{i\in I_k}\psi_{m_r}(X_i)$. 4) Aggregate these estimators:
$$\hat{\theta}^K=\frac{1}{K}\sum\limits_{k=1}^K\hat{\theta}_k.$$
5) For $r=1,\dots,d$ construct the uniform valid confidence interval
$$\left[\hat{\theta}_{m_r}^K-\frac{c_\alpha\hat{\sigma}_{m_r}^K}{n},\hat{\theta}_{m_r}^K+\frac{c_\alpha\hat{\sigma}_{m_r}^K}{n}\right]$$
with
\begin{align*}
\hat{J}_{m_r}^K &=-\frac{1}{K}\sum\limits_{k=1}^K(X_k(X_k-\hat{\eta}_{k,m_r}^{(2)}X_{-m_r})),\\
\hat{\sigma}_{m_r}^K&=\sqrt{(\hat{J}_{m_r}^K)^{-2}\frac{1}{K}\sum\limits_{k=1}^K\left(\psi_{m_r}^2(X,\hat{\theta}^K_{m_r},\hat{\eta}_{k,m_r})\right)}.
\end{align*}
$c_\alpha$ is the $1-\alpha$ bootstrap quantile of $\sup\limits_{r=1,\dots,d}\hat{\mathcal{N}}_{m_r}$ with
$$\hat{\mathcal{N}}_{m_r}=\frac{1}{\sqrt{n}}\sum\limits_{i=1}^n\xi_i\hat{\psi}^K_{m_r}\big(X^{(i)}\big)$$
where $(\xi_i)_{i=1}^n$ are independent standard normal random variables which are independent from $\big(X^{(i)}\big)_{i=1}^n$ and
$$\hat{\psi}^K_{m_r}(X):=-\left(\hat{\sigma}_{m_r}^K\hat{J}_{m_r}^K\right)^{-1}\psi_{m_r}(X,\hat{\theta}^K_{m_r},\hat{\eta}^K_{m_r}).$$
\end{algorithm}
\noindent
The confidence region above corresponds to (\ref{cube}). Confidence regions corresponding to (\ref{sphere}) or (\ref{S-sparse}) can be constructed in an analogous way.


\section{Simulation Study}\label{simulation}
This section provides a simulation study on the proposed method. In each example the precision matrix of the Gaussian graphical model is generated as in the $R$-package $huge$ \cite{zhao2012huge}. Hence, the corresponding adjacency matrix $A$ is generated by setting the nonzero off-diagonal elements to be one and each other element to be zero. To obtain a positive definite pre-version of the precision matrix we set
$$\Phi_{pre}:= v\cdot A+(|\Lambda_{\min}(v\cdot A)|+0.1+u)\cdot I_{p\times p}.$$
Here $v=0.3$ and $u=0.1$ are chosen to control the magnitude of partial correlations.
The covariance matrix $\Sigma$ is generated by inverting $\Phi_{pre}$ and scaling the variances to one. The corresponding precision matrix $\Phi$ is given by $\Sigma^{-1}$. For a given $p$ we generate $n=200$ independent samples of 
$$X=(X_1,\dots,X_p)\sim\mathcal{N}(0,\Sigma)$$
and evaluate whether our test statistic would reject the null hypothesis for a specific set of edges $\mathcal{M}$ which satisfies the null hypothesis. Finally the acceptance rate is calculated over $l=1000$ independent simulations for a given confidence level $1-\alpha=0.95$. 
\subsection{Simulation settings} In our simulation study we estimate the correlation structure of four different designs that are described in the following.
\subsubsection{Example 1: Random Graph} Each pair of off-diagonal elements of the covariance matrix of the first $p-1$ regressors is randomly set to non-zero with probability $prob = 5/p$. The last regressor is added as an independent random variable. It results in about $(p-1)\cdot(p-2)\cdot prob /2$ edges in the graph. The corresponding precision matrix is of the form 
$$
\Phi:=\left(
\begin{array}{cc}
   \raisebox{-15pt}{\textrm{\huge \mbox{{$B$}}}}&0 \\[-4ex]
   & \vdots\\
   & 0\\
  0 \cdots 0 & 1
\end{array}
\right)
$$
where $B$ is a sparse matrix. We test the hypothesis, whether the last regressor is independent from all other regressors, corresponding to 
$$\mathcal{M}=\{(p,1),\dots,(p,p-1)\}.$$
\subsubsection{Example 2: Cluster Graph}
The regressors are evenly partitioned into $g=4$  disjoint groups. Each pair of off-diagonal elements $\Phi_{(i,j)}$ is set non-zero with probability $prob=5/p$, if both $i$ and $j$ belong to the same group. It results in about $g\cdot(p/g)\cdot(p/g-1)\cdot prob/2$ edges in the graph. The precision Matrix is of the form 
$$
\Phi:=\left(
\begin{array}{cccc}
B_1&&&{\textrm{\huge \mbox{{$0$}}}}\\
&B_2&&\\
&&B_3&\\
{\textrm{\huge \mbox{{$0$}}}}&&&B_4
\end{array}
\right)
$$
where each block $B_i$ is a sparse matrix.
We test the hypothesis that the first two hubs are conditionally independent. This corresponds to testing the tuples 
$$\mathcal{M}=\{(1,p/4+1),\dots,(1,p/2),(2,p/4+1),\dots,(p/4,p/2)\}.$$
\begin{figure}[H]
    \centering
    \subfloat[Random Graph]{{\includegraphics[width=5cm]{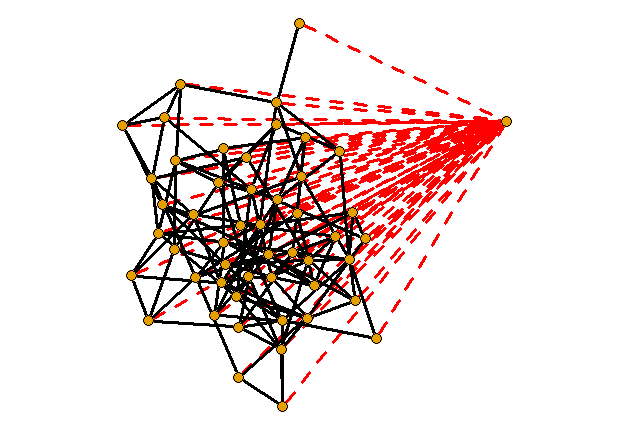} }}%
    \qquad
    \subfloat[Cluster Graph]{{\includegraphics[width=5cm]{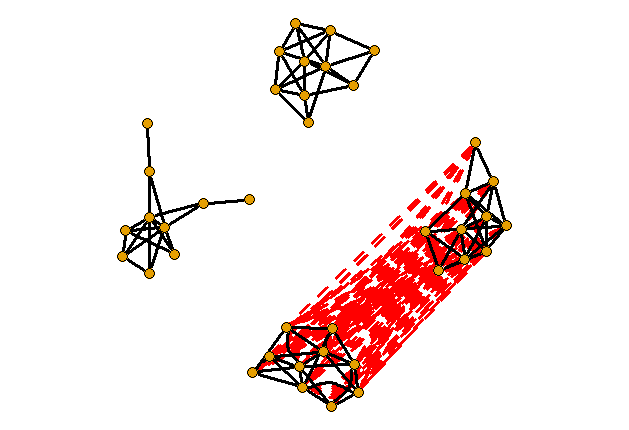} }}%
    \caption{Examples}%
		\scriptsize{The edges of the graph are colored in black and the edges contained in the hypothesis in red.}
    \label{fig:example}%
\end{figure}

\subsubsection{Example 3: Approximately Sparse Random Graph}
In this example we generate a random graph structure as in example $1$, but instead of setting the other elements of the adjacency matriy $A$ to zero we generate independent random entries from a uniform distribution on $[-a,a]$ with $a=1/20$. This results in a precision matrix of the form 
$$
\Phi:=\left(
\begin{array}{cc}
   \raisebox{-15pt}{\textrm{\huge \mbox{{$B$}}}}&0 \\[-4ex]
   & \vdots\\
   & 0\\
  0 \cdots 0 & 1
\end{array}
\right)
$$
where $B$ is not a sparse matrix anymore. We then again test the hypothesis, whether the last regressor is independent from all other regressors, corresponding to 
$$\mathcal{M}=\{(p,1),\dots,(p,p-1)\}.$$ 
\subsubsection{Example 4: Independent Graph}
By setting
$$\Phi:=I_{p\times p}$$
we generate samples of $p$ independent normal distributed random variables. We can test the hypothesis whether the regressors are independent by choosing
$$\mathcal{M}=\{(1,2),\dots,(1,p),(2,3),\dots,(p-1,p)\}.$$
\noindent
\subsection{Simulation results}
We provide simulated acceptance rates of our proposed estimation procedure with $B=1000$ bootstrap samples for all of the examples above. Confidence Intervall I corresponds to the standard case in (\ref{cube}), whereas Confidence Intervall II is based on the approximation of the sphere in (\ref{sphere}). In summary, the results reveal that the empirical acceptance rate is, on average, close to the nominal level of $95\%$ with a mean absolute deviation of $2.581\%$ over all simulations. The Confidence Intervall II has got a mean absolute deviation of $1.875\%$ and performs significantly better than Confidence Intervall I with a mean absolute deviation of $3.287\%$. More complex S-sparsely convex sets seem to result in better acceptance rates, whereas higher exponents do not improve the rates. The lowest mean absolute deviation ($1.138\%$) is achieved in table 2 for $S=5$, $exp=1$ and without cross-fitting.  
\begin{table}[H]
\centering
\begin{tabular}{ccccccccc}
  \toprule
	& & &\multicolumn{3}{c}{Confidence Interval I}&\multicolumn{3}{c}{Confidence Intervall II}\\
	\cmidrule(l{5pt}r{5pt}){4-6} \cmidrule(l{5pt}r{5pt}){7-9}
Model & p & d  & lasso & post-lasso & sqrt-lasso & lasso & post-lasso & sqrt-lasso \\ 
  \midrule
\multirow{3}{*}{random}&20 & 19 & 0.931 & 0.938 & 0.936 & 0.929 & 0.930 & 0.935 \\ 
  &50 & 49 & 0.915 & 0.915 & 0.916 & 0.926 & 0.929 & 0.932 \\ 
  &100 & 99 & 0.912 & 0.912 & 0.908 & 0.927 & 0.930 & 0.929 \\
	\midrule
	\multirow{3}{*}{cluster}&20 & 25 & 0.916 & 0.942 & 0.918 & 0.915 & 0.930 & 0.921 \\ 
  &40 & 100 & 0.916 & 0.919 & 0.917 & 0.934 & 0.947 & 0.937 \\ 
  &60 & 225 & 0.897 & 0.893 & 0.899 & 0.921 & 0.922 & 0.927 \\ 
	\midrule
	\multirow{3}{*}{approx}&20 & 19 & 0.931 & 0.931 & 0.931 & 0.947 & 0.946 & 0.947 \\ 
  &50 & 49 & 0.908 & 0.908 & 0.908 & 0.920 & 0.920 & 0.920 \\ 
  &100 & 99 & 0.902 & 0.902 & 0.902 & 0.935 & 0.935 & 0.935 \\ 
	\midrule
	\multirow{3}{*}{indepent}&5 & 10 & 0.931 & 0.931 & 0.931 & 0.933 & 0.933 & 0.933 \\ 
  &10 & 45 & 0.927 & 0.927 & 0.927 & 0.937 & 0.937 & 0.937 \\ 
  &20 & 190 & 0.896 & 0.896 & 0.896 & 0.920 & 0.920 & 0.920 \\ 
   \bottomrule
\end{tabular}
\caption{Simulation results for S=1,exp=1 and 1-fold} 
\end{table}

\begin{table}[H]
\centering
\begin{tabular}{ccccccccc}
  \toprule
	& & &\multicolumn{3}{c}{Confidence Interval I}&\multicolumn{3}{c}{Confidence Intervall II}\\
	\cmidrule(l{5pt}r{5pt}){4-6} \cmidrule(l{5pt}r{5pt}){7-9}
Model & p & d  & lasso & post-lasso & sqrt-lasso & lasso & post-lasso & sqrt-lasso \\ 
  \midrule
\multirow{3}{*}{random}&20 & 19 & 0.969 & 0.925 & 0.956 & 0.951 & 0.932 & 0.947 \\ 
  &50 & 49 & 0.942 & 0.944 & 0.944 & 0.942 & 0.954 & 0.953 \\ 
  &100 & 99 & 0.934 & 0.941 & 0.940 & 0.950 & 0.949 & 0.952 \\ 
	\midrule
	\multirow{3}{*}{cluster}&20 & 25 & 0.972 & 0.958 & 0.973 & 0.914 & 0.936 & 0.914 \\ 
  &40 & 100 & 0.941 & 0.937 & 0.945 & 0.930 & 0.936 & 0.942 \\ 
  &60 & 225 & 0.931 & 0.947 & 0.942 & 0.943 & 0.937 & 0.950 \\  
	\midrule
	\multirow{3}{*}{approx}&20 & 19 & 0.958 & 0.958 & 0.958 & 0.965 & 0.965 & 0.965 \\ 
  &50 & 49 & 0.937 & 0.937 & 0.937 & 0.940 & 0.940 & 0.940 \\ 
  &100 & 99 & 0.920 & 0.921 & 0.920 & 0.936 & 0.936 & 0.936 \\ 
	\midrule
	\multirow{3}{*}{indepent}&5 & 10 & 0.951 & 0.951 & 0.951 & 0.951 & 0.951 & 0.951 \\ 
  &10 & 45 & 0.932 & 0.932 & 0.932 & 0.952 & 0.952 & 0.952 \\ 
  &20 & 190 & 0.926 & 0.926 & 0.926 & 0.947 & 0.947 & 0.947 \\
   \bottomrule
\end{tabular}
\caption{Simulation results for S=5,exp=1 and 1-fold} 
\end{table}

\begin{table}[H]
\centering
\begin{tabular}{ccccccccc}
  \toprule
	& & &\multicolumn{3}{c}{Confidence Interval I}&\multicolumn{3}{c}{Confidence Intervall II}\\
	\cmidrule(l{5pt}r{5pt}){4-6} \cmidrule(l{5pt}r{5pt}){7-9}
Model & p & d  & lasso & post-lasso & sqrt-lasso & lasso & post-lasso & sqrt-lasso \\ 
  \midrule
\multirow{3}{*}{random}&20 & 19 & 0.909 & 0.916 & 0.921 & 0.916 & 0.921 & 0.930 \\ 
  &50 & 49 & 0.931 & 0.910 & 0.926 & 0.926 & 0.907 & 0.927 \\ 
  &100 & 99 & 0.907 & 0.909 & 0.909 & 0.917 & 0.934 & 0.923 \\ 
	\midrule
	\multirow{3}{*}{cluster}&20 & 25 & 0.910 & 0.905 & 0.905 & 0.904 & 0.898 & 0.901 \\ 
  &40 & 100 & 0.909 & 0.910 & 0.910 & 0.905 & 0.919 & 0.921 \\ 
  &60 & 225 & 0.885 & 0.894 & 0.898 & 0.912 & 0.925 & 0.934 \\
	\midrule
	\multirow{3}{*}{approx}&20 & 19 & 0.929 & 0.928 & 0.929 & 0.929 & 0.928 & 0.929 \\ 
  &50 & 49 & 0.888 & 0.888 & 0.888 & 0.911 & 0.911 & 0.911 \\ 
  &100 & 99 & 0.907 & 0.907 & 0.907 & 0.936 & 0.936 & 0.936 \\ 
	\midrule
	\multirow{3}{*}{indepent}&5 & 10 & 0.930 & 0.930 & 0.930 & 0.939 & 0.939 & 0.939 \\ 
  &10 & 45 & 0.921 & 0.921 & 0.921 & 0.933 & 0.933 & 0.933 \\ 
  &20 & 190 & 0.916 & 0.916 & 0.916 & 0.938 & 0.938 & 0.938 \\ 
   \bottomrule
\end{tabular}
\caption{Simulation results for S=5,exp=2 and 1-fold} 
\end{table}

\begin{table}[H]
\centering
\begin{tabular}{ccccccccc}
  \toprule
	& & &\multicolumn{3}{c}{Confidence Interval I}&\multicolumn{3}{c}{Confidence Intervall II}\\
	\cmidrule(l{5pt}r{5pt}){4-6} \cmidrule(l{5pt}r{5pt}){7-9}
Model & p & d  & lasso & post-lasso & sqrt-lasso & lasso & post-lasso & sqrt-lasso \\ 
  \midrule
\multirow{3}{*}{random}&20 & 19 & 0.917 & 0.912 & 0.919 & 0.919 & 0.932 & 0.918 \\ 
  &50 & 49 & 0.927 & 0.911 & 0.925 & 0.938 & 0.936 & 0.938 \\ 
  &100 & 99 & 0.903 & 0.894 & 0.907 & 0.926 & 0.933 & 0.927 \\  
	\midrule
	\multirow{3}{*}{cluster}&20 & 25 & 0.920 & 0.899 & 0.918 & 0.930 & 0.929 & 0.929 \\ 
  &40 & 100 & 0.920 & 0.883 & 0.919 & 0.927 & 0.926 & 0.923 \\ 
  &60 & 225 & 0.889 & 0.885 & 0.896 & 0.920 & 0.930 & 0.928 \\ 
	\midrule
	\multirow{3}{*}{approx}&20 & 19 & 0.921 & 0.922 & 0.921 & 0.932 & 0.934 & 0.932 \\ 
  &50 & 49 & 0.899 & 0.899 & 0.899 & 0.926 & 0.926 & 0.926 \\ 
  &100 & 99 & 0.889 & 0.889 & 0.889 & 0.930 & 0.929 & 0.930 \\ 
	\midrule
	\multirow{3}{*}{indepent}&5 & 10 & 0.922 & 0.923 & 0.922 & 0.935 & 0.934 & 0.935 \\ 
  &10 & 45 & 0.905 & 0.905 & 0.905 & 0.937 & 0.937 & 0.937 \\ 
  &20 & 190 & 0.903 & 0.903 & 0.903 & 0.936 & 0.936 & 0.936 \\ 
   \bottomrule
\end{tabular}
\caption{Simulation results for S=1,exp=1 and 3-fold} 
\end{table}

\begin{table}[H]
\centering
\begin{tabular}{ccccccccc}
  \toprule
	& & &\multicolumn{3}{c}{Confidence Interval I}&\multicolumn{3}{c}{Confidence Intervall II}\\
	\cmidrule(l{5pt}r{5pt}){4-6} \cmidrule(l{5pt}r{5pt}){7-9}
Model & p & d  & lasso & post-lasso & sqrt-lasso & lasso & post-lasso & sqrt-lasso \\ 
  \midrule
\multirow{3}{*}{random}&20 & 19 & 0.970 & 0.919 & 0.964 & 0.950 & 0.932 & 0.958 \\ 
  &50 & 49 & 0.923 & 0.911 & 0.927 & 0.938 & 0.951 & 0.935 \\ 
  &100 & 99 & 0.929 & 0.925 & 0.930 & 0.949 & 0.940 & 0.948 \\ 
	\midrule
	\multirow{3}{*}{cluster}&20 & 25 & 0.971 & 0.970 & 0.971 & 0.915 & 0.931 & 0.915 \\ 
  &40 & 100 & 0.926 & 0.915 & 0.925 & 0.925 & 0.917 & 0.924 \\ 
  &60 & 225 & 0.923 & 0.925 & 0.926 & 0.917 & 0.939 & 0.930 \\
	\midrule
	\multirow{3}{*}{approx}&20 & 19 & 0.959 & 0.959 & 0.959 & 0.958 & 0.956 & 0.958 \\ 
  &50 & 49 & 0.932 & 0.932 & 0.932 & 0.931 & 0.933 & 0.931 \\ 
  &100 & 99 & 0.929 & 0.929 & 0.929 & 0.949 & 0.950 & 0.949 \\ 
	\midrule
	\multirow{3}{*}{indepent}&5 & 10 & 0.940 & 0.940 & 0.940 & 0.951 & 0.951 & 0.951 \\ 
  &10 & 45 & 0.922 & 0.922 & 0.922 & 0.938 & 0.938 & 0.938 \\ 
  &20 & 190 & 0.930 & 0.930 & 0.930 & 0.938 & 0.938 & 0.938 \\ 
   \bottomrule
\end{tabular}
\caption{Simulation results for S=5,exp=1 and 3-fold} 
\end{table}

\begin{table}[H]
\centering
\begin{tabular}{ccccccccc}
  \toprule
	& & &\multicolumn{3}{c}{Confidence Interval I}&\multicolumn{3}{c}{Confidence Intervall II}\\
	\cmidrule(l{5pt}r{5pt}){4-6} \cmidrule(l{5pt}r{5pt}){7-9}
Model & p & d  & lasso & post-lasso & sqrt-lasso & lasso & post-lasso & sqrt-lasso \\ 
  \midrule
\multirow{3}{*}{random}&20 & 19 & 0.914 & 0.897 & 0.918 & 0.922 & 0.921 & 0.923 \\ 
  &50 & 49 & 0.914 & 0.896 & 0.911 & 0.920 & 0.920 & 0.921 \\ 
  &100 & 99 & 0.891 & 0.878 & 0.893 & 0.918 & 0.909 & 0.917 \\ 
	\midrule
	\multirow{3}{*}{cluster}&20 & 25 & 0.885 & 0.882 & 0.888 & 0.900 & 0.896 & 0.901 \\ 
  &40 & 100 & 0.880 & 0.877 & 0.879 & 0.898 & 0.910 & 0.907 \\ 
  &60 & 225 & 0.886 & 0.884 & 0.897 & 0.915 & 0.921 & 0.932 \\
	\midrule
	\multirow{3}{*}{approx}&20 & 19 & 0.931 & 0.930 & 0.931 & 0.938 & 0.937 & 0.938 \\ 
  &50 & 49 & 0.914 & 0.913 & 0.914 & 0.932 & 0.933 & 0.932 \\ 
  &100 & 99 & 0.894 & 0.894 & 0.894 & 0.924 & 0.924 & 0.924 \\  
	\midrule
	\multirow{3}{*}{indepent}&5 & 10 & 0.923 & 0.922 & 0.923 & 0.943 & 0.942 & 0.943 \\ 
  &10 & 45 & 0.917 & 0.916 & 0.917 & 0.934 & 0.935 & 0.934 \\ 
  &20 & 190 & 0.890 & 0.890 & 0.890 & 0.932 & 0.932 & 0.932 \\ 
   \bottomrule
\end{tabular}
\caption{Simulation results for S=5,exp=2 and 3-fold} 
\end{table}

\newpage
\appendix
\section{Proof of Theorem 1}\label{appendixmain}
\begin{proof}
We want to use  corollary $2.2$ from Belloni et al. (2018) \cite{belloni2018uniformly}. Consequently, we will show that their assumptions 2.1-2.4 and the growth conditions of corollary $2.2$ hold by modifying the proof of corollary $3.2$ in \cite{belloni2018uniformly}. To make the proof more comparable we try to keep the notation as similar as possible. This implies that we use $C$ for a strictly positive constant, independent of $n$ and $r$, which may have a different value in each appearance. The notation $a_n\lesssim b_n$ stands for $a_n\le Cb_n$ for all $n$ for some fixed $C$. Additionally $a_n=o(1)$ stands for uniform convergence towards zero meaning there exists sequence $(b_n)_{n\ge 1}$ with $|a_n|\le b_n$, $b_n$ is independent of $P\in\mathcal{P}_n$ for all $n$ and $b_n\to 0$. Finally, the notation $a_n\lesssim_{P}b_n$ means that for any $\epsilon>0$, there exists $C$ such that uniformly over all $n$ we have $P_P(a_n>Cb_n)\le \epsilon$.\\ 
Let $m_r=(j,k)$ be an arbitrary set in $\mathcal{M}$. We have
$$\max\limits_r\E\left[\left(\nu^{(m_r)}\right)^2\right]\lesssim 1\ \text{and}\ \max\limits_r\E\left[\left(\varepsilon^{(m_r)}\right)^2\right]\lesssim 1 $$
due to the assumptions \ref{A3} and \ref{A4}. Define the convex set
$$T_{m_r}=\{\eta=(\eta^{(1)},\eta^{(2)}):\eta^{(1)}\in\mathbb{R}^{p-2},\eta^{(2)}\in\mathbb{R}^{p-2}\}$$
and endow $T_{m_r}$ with the norm
$$||\eta||_e=||\eta^{(1)}||_2\vee ||\eta^{(2)}||_2.$$
Further let $\tau_n:=\sqrt{\frac{s\log (a_n)}{n}}$ and define the nuisance realization set
\begin{align*}\mathcal{T}_{m_r}=\bigg\{&\eta\in T_{m_r}:||\eta^{(1)}||_0\vee ||\eta^{(2)}||_0\le Cs,\\
&||\eta^{(1)}-\beta^{(m_r)}||_2\vee ||\eta^{(2)}-\gamma^{(m_r)}||_2\le C\tau_n,\\
&||\eta^{(1)}-\beta^{(m_r)}||_1\vee ||\eta^{(2)}-\gamma^{(m_r)}||_1\le C\sqrt{s}\tau_n\bigg\}\cup\left\{\left(\beta^{(m_r)},\gamma^{(m_r)}\right)\right\}
\end{align*}
for a sufficiently large constant $C>0$. First we verify Assumption 2.1 (i). The moment condition holds since
\begin{align*}
&\quad\E[\psi_{m_r}(X,\theta_{m_r},\eta_{m_r})]
\\&=\E[\varepsilon^{(m_r)}\nu^{(m_r)}]
\\&=\E[\E[\varepsilon^{(m_r)}\nu^{(m_r)}|X_{-j}]]=\E[\nu^{(m_r)}\underbrace{\E[\varepsilon^{(m_r)}|X_{-j}]}_{=0}]=0.
\end{align*}
In addition, we have
\begin{align*}
S_n:&=\E\left[\max\limits_r|\sqrt{n}\E_n[\psi_{m_r}(X,\theta_{m_r},\eta_{m_r})]|\right]\\
&=\E\left[\sup\limits_{f\in\mathcal{F}} \mathbb{G}_n(f)\right]
\end{align*}
with $\mathcal{F}=\{\varepsilon^{(m_r)}\nu^{(m_r)}|r=1,\dots,d\}$ and $ \mathbb{G}_n(f):=\sqrt{n}|\E_n[f]-\E[f]|$. By the same arguments as in the beginning of proof of theorem \ref{uniformlasso} we conclude that the envelope $\sup\limits_{f\in\mathcal{F}}|f|$ of $\mathcal{F}$ fulfills
\begin{align*}
||\max\limits_r|\varepsilon^{(m_r)}\nu^{(m_r)}|||_{P,q}&=\E\left[\max\limits_r\left(|\varepsilon^{(m_r)}\nu^{(m_r)}|\right)^q\right]^{1/q}\\
&\le\E\left[\max\limits_r\left(|\varepsilon^{(m_r)}|\right)^{2q}\right]^{1/2q}\E\left[\max\limits_r\left(|\nu^{(m_r)}|\right)^{2q}\right]^{1/2q}\\
&\le C\log(d),
\end{align*}
since the error terms are normal distributed.
Using lemma O.2 (Maximal Inequality I) in \cite{belloni2018uniformly} with $|\mathcal{F}|=d$, we have 
\begin{align*}
S_n\le C\log^{1/2}(d)+C\log^{1/2}(d)\left(n^{\frac{2}{q}}\frac{\log^{3}(d)}{n}\right)^{1/2}\lesssim\log^{1/2}(d)
\end{align*}
by the assumption \ref{A2} for a $q>2\tilde{q}$. Hence, assumption \ref{A3} implies that for all $r=1,\dots,d$, $\Theta_{m_r}$ contains an interval of radius $Cn^{-\frac{1}{2}}S_n\log(n)$ centered at $\theta_{m_r}$ for all sufficiently large $n$ for any constant $C$. Assumption 2.1 (i) follows.\\ \\
For all $m_r\in\mathcal{M}$, the map $(\theta,\eta)\mapsto\psi_{m_r}(X,\theta,\eta)$ is twice continuously Gateaux-differentiable on $\Theta_{m_r}\times\mathcal{T}_{m_r}$, and so is the map $(\theta,\eta)\mapsto\E[\psi_{m_r}(X,\theta,\eta)]$. Further we have
\begin{align*}
D_{m_r,0}[\eta,\eta_{m_r}]:&=\partial_t\E[\psi_{m_r}(X,\theta_{m_r},\eta_{m_r}+t(\eta-\eta_{m_r}))]\big|_{t=0}\\
&=\E\bigg[\partial_t\bigg\{\left(X_j-\theta_{m_r}X_k-\left(\eta_{m_r}^{(1)}+t(\eta^{(1)}-\eta_{m_r}^{(1)})\right) X_{-m_r}\right)\\
&\qquad\left(X_k-\left(\eta_{m_r}^{(2)}+t(\eta^{(2)}-\eta_{m_r}^{(2)})\right) X_{-m_r}\right)\bigg\}\bigg]\big|_{t=0}\\
&=\E[\varepsilon^{(m_r)}(\eta^{(2)}_{m_r}-\eta^{(2)})X_{-m_r}]+\E[(\eta_{m_r}^{(1)}-\eta^{(1)})X_{-m_r}\nu^{(m_r)}]\\
&=0.
\end{align*}
Therefore, Assumptions 2.1 (ii) and 2.1 (iii) hold. Remark that
 \begin{align*}
|J_{m_r}|&=|\partial_\theta\E[\psi_{m_r}(X,\theta,\eta_{m_r})]|_{\theta=\theta_{m_r}}|\\
&=|\E[-X_k\nu^{(m_r)}]||=|\E[(\nu^{(m_r)})^2]|\le C
\end{align*}
and
\begin{align*}
|J_{m_r}|&=|\E[(\nu^{(m_r)})^2]|\ge c
\end{align*}
due to assumption \ref{A4}. Since the score $\psi$ is linear with respect to $\theta$, we have for all $m_r\in\mathcal{M}$ and $\theta\in\Theta_{m_r}$
\begin{align*}
\E[\psi_{m_r}(X,\theta,\eta_{m_r})]=J_{m_r}(\theta-\theta_{m_r})
\end{align*}
using the moment condition. This gives us Assumption 2.1 (iv).\\
For all $t\in [0,1)$, $m_r\in\mathcal{M}$, $\theta\in\Theta_{m_r}$, $\eta\in \mathcal{T}_{m_r}$ we have
\begin{align*}
&\E \big[(\psi_{m_r}(X,\theta,\eta)-\psi_{m_r}(X,\theta_{m_r},\eta_{m_r}))^2\big]\\
=\ &\E \big[(\psi_{m_r}(X,\theta,\eta)-\psi_{m_r}(X,\theta_{m_r},\eta)+\psi_{m_r}(X,\theta_{m_r},\eta)-\psi_{m_r}(X,\theta_{m_r},\eta_{m_r}))^2\big]\\
\le\ &C\bigg(\underbrace{\E \big[(\psi_{m_r}(X,\theta,\eta)-\psi_{m_r}(X,\theta_{m_r},\eta))^2\big]}_{=:I}\\
&\quad\vee \underbrace{\E \big[(\psi_{m_r}(X,\theta_{m_r},\eta)-\psi_{m_r}(X,\theta_{m_r},\eta_{m_r}))^2\big]}_{=:II}\bigg)
\end{align*}
with
\begin{align*}
I&=|\theta-\theta_{m_r}|^2\E \left[\left(X_k(X_k-\eta^{(2)}X_{-m_r})\right)^2\right]\\
&\le |\theta-\theta_{m_r}|^2\left(\E[X_k^2]E[(X_k-\eta^{(2)}X_{-m_r})^2]\right)^{1/2}\\
&\le C |\theta-\theta_{m_r}|^2
\end{align*}
due to assumptions \ref{A3}, \ref{A4} and the definition of $\mathcal{T}_{m_r}$. Additionally we have 
\begin{align*}
II&=\E \bigg[\Big(\big(X_j-\theta_{m_r}X_k-\eta^{(1)}X_{-m_r}\big)\big(X_k-\eta^{(2)}X_{-m_r}\big)\\
&\quad\quad-\big(X_j-\theta_{m_r}X_k-\eta^{(1)}_{m_r}X_{-m_r}\big)\big(X_k-\eta^{(2)}_{m_r}X_{-m_r}\big)\Big)^2\bigg]\\
&=\E \bigg[\Big(\big(X_j-\theta_{m_r}X_k-\eta^{(1)}X_{-m_r}\big)\big((\eta^{(2)}_{m_r}-\eta^{(2)})X_{-m_r}\big)\\
&\quad\quad+\big(X_k-\eta^{(2)}_{m_r}X_{-m_r}\big)\big((\eta^{(1)}_{m_r}-\eta^{(1)})X_{-m_r}\big)\Big)^2\bigg]\\
&\le C \Big(\|\eta^{(2)}_{m_r}-\eta^{(2)}\|_2\vee \|\eta^{(1)}_{m_r}-\eta^{(1)}\|_2\Big)^2\\
&=C \|\eta_{m_r}-\eta\|_e^2
\end{align*}
with similar arguments as in $I$ above using
$$ \sup\limits_{\|\xi\|_2=1} \E\left[(\xi X)^4\right]\le C$$
due to the normal distributed design. Combining these results gives us Assumption 2.1 (v) (a).\\
Observe that
\begin{align*}
&\ \Big|\partial_t\E \Big[\psi_{m_r}\big(X,\theta,\eta_{m_r}+t(\eta-\eta_{m_r})\big)\Big]\Big|\\
=&\ \bigg|\E \bigg[\Big(X_j-\theta X_k-\big(\eta^{(1)}_{m_r}+t(\eta^{(1)}-\eta^{(1)}_{m_r})\big)X_{-m_r}\Big)\big((\eta^{(2)}_{m_r}-\eta^{(2)})X_{-m_r}\big)\\
&\quad\quad +\Big(X_k-\big(\eta^{(2)}_{m_r}+t(\eta^{(2)}-\eta^{(2)}_{m_r})\big)X_{-m_r}\Big)\big((\eta^{(1)}_{m_r}-\eta^{(1)})X_{-m_r}\big)\bigg]\bigg|\\
\le&\ C \|\eta_{m_r}-\eta\|_e
\end{align*}
with the same argument as above, which gives us Assumption 2.1 (v) (b) with $B_{1n}=C$. To complete the Assumption 2.1 (v) (c) with $B_{2n}=C$ observe that
\begin{align*}
&\Big|\partial^2_t\E \Big[\psi_{m_r}\big(X,\theta_{m_r}+t(\theta-\theta_{m_r}),\eta_{m_r}+t(\eta-\eta_{m_r})\big)\Big]\Big|\\
=& \bigg|\partial_t\E \bigg[\Big(X_j-\big(\theta_{m_r}+t(\theta-\theta_{m_r})\big)X_k-\big(\eta^{(1)}_{m_r}+t(\eta^{(1)}-\eta^{(1)}_{m_r})\big)X_{-m_r}\Big)\\
&\quad\quad\cdot\big((\eta^{(2)}_{m_r}-\eta^{(2)})X_{-m_r}\big)\\
&\quad\quad +\Big(X_k-\big(\eta^{(2)}_{m_r}+t(\eta^{(2)}-\eta^{(2)}_{m_r})\big)X_{-m_r}\Big)\\
&\quad\quad\cdot\big((\theta_{m_r}-\theta)X_k+(\eta^{(1)}_{m_r}-\eta^{(1)})X_{-m_r}\big)\bigg]\bigg|\\
=&\Big|2\E \Big[\big((\eta^{(2)}_{m_r}-\eta^{(2)})X_{-m_r}\big)\big((\theta_{m_r}-\theta)X_k+(\eta^{(1)}_{m_r}-\eta^{(1)})X_{-m_r}\big)\Big]\Big|\\
\le&\ 2 \bigg(\underbrace{\E \Big[\big((\eta^{(2)}_{m_r}-\eta^{(2)})X_{-m_r}\big)^2\Big]}_{\le C\|\eta^{(2)}_{m_r}-\eta^{(2)}\|_2^2}\underbrace{\E \Big[\big((\theta_{m_r}-\theta)X_k+(\eta^{(1)}_{m_r}-\eta^{(1)})X_{-m_r}\big)^2\Big]}_{\le C\big(|\theta_{m_r}-\theta|^2+\|\eta^{(1)}_{m_r}-\eta^{(1)}\|_2^2\big)}\bigg)^{1/2}\\
\le &\ C\big(|\theta_{m_r}-\theta|^2\vee\|\eta_{m_r}-\eta\|_e^2\big).
\end{align*}
Therefore Assumption 2.1 holds. Due to the construction of $\mathcal{T}_{m_r}$ Assumptions 2.2 (ii) and (iii) hold. Next, we show that the assumptions of theorem \ref{uniformlasso} from section \ref{uniformestimation} hold which implies Assumption 2.2 (i). Remark that conditions \ref{tails} and \ref{growthc} are satisfied with $\rho=2$. Condition \ref{A1} implies condition \ref{asparse}. Let $\underline{\sigma}^2>0$ be a uniform lower bound for the variances of the error terms and the regressors and let $c:=\underline{\sigma}z_{\tilde{c}}$, where $z_{\tilde{c}}$ is the $\tilde{c}$-quantile of a standard normal distribution for an arbitrary but fixed $\tilde{c}\in(\frac{1}{2},\frac{3}{4})$. Uniformly for all $r=1,\dots,d$ and $l\in \{1,\dots,p\}\setminus \{j\}$, it holds
\begin{align*}
P\left((\varepsilon^{(m_r)})^2X_{l}^2\ge c^4\right)
&=1-P\left(|\varepsilon^{(m_r)} X_{l}|\le c^2\right)\\
&\ge 1-P\left(|\varepsilon^{(m_r)}|\le c\vee|X_{l}|\le c\right)\\
&\ge 1-\left(P\left(|\varepsilon^{(m_r)}|\le c\right)+P\left(|X_{l}|\le c\right)\right)\\
&\ge 1-2P\left(\underline{\sigma}|Z|\le c\right)\\
&=3-4\tilde{c}>0
\end{align*}
where $Z\sim\mathcal{N}(0,1)$, which implies that
\begin{align*}
\min_{r}\min_{l}\E[(\varepsilon^{(m_r)})^2X_{l}^2]&\ge c^4 (3-4\tilde{c})>0.
\end{align*}
Analogously 
\begin{align*}
\min_{r}\min_{l}\E[(\nu^{(m_r)})^2X_{l}^2]>0.
\end{align*}
Combined with condition \ref{A4} this implies condition \ref{eigen}. Therefore we are able to estimate the nuisance parameters at a sufficiently fast rate.\\
Define
\begin{align*}
\mathcal{F}_{1}:&=\Big\{\psi_{m_r}(\cdot,\theta,\eta):r\in\{1,\dots,d\},\theta\in\Theta_{m_r},\eta\in\mathcal{T}_{m_r}\Big\}.
\end{align*}
To bound the covering entropy of $\mathcal{F}_{1}$ we at first exclude the true nuisance parameter and define
\begin{align*}
\mathcal{F}_{1,1}:&=\Big\{\psi_{m_r}(\cdot,\theta,\eta):r\in\{1,\dots,d\},\theta\in\Theta_{m_r},\eta\in\mathcal{T}_{m_r}\setminus\{\eta_{m_r}\}\Big\}\subseteq\mathcal{F}_{1,1}^{(1)}\mathcal{F}_{1,1}^{(2)}
\end{align*}
with 
\begin{align*}
\mathcal{F}_{1,1}^{(1)}&=\{X\to (X_j-\theta X_k-\eta^{(1)} X_{-m_r}):r\in\{1,\dots,d\},\theta\in\Theta_{m_r},\eta^{(1)}\in\mathcal{T}^*_{m_r,1}\}\\
\mathcal{F}_{1,1}^{(2)}&=\{X\to (X_k-\eta^{(2)} X_{-m_r}):r\in\{1,\dots,d\},\eta^{(2)}\in\mathcal{T}^*_{m_r,2}\}
\end{align*}
where $\mathcal{T}^*_{m_r}:=\mathcal{T}_{m_r}\setminus\{\eta_{m_r}\}$. Observe that the envelope $F_{1,1}^{(1)}$ of $\mathcal{F}_{1,1}^{(1)}$ fulfills
\begin{align*}
\big\|\big(F_{1,1}^{(1)}\big)^2\big\|_{P,2q}&\le \Big\|\sup\limits_{r\in\{1,\dots,d\}}\sup\limits_{\theta\in\Theta_{m_r},\|\eta^{(1)}_{m_r}-\eta^{(1)}\|_1\le C\sqrt{s}\tau_n}\Big(|\varepsilon^{(m_r)}|\\
&\quad\quad +|(\theta_{m_r}-\theta)X_k|+|(\eta^{(1)}_{m_r}-\eta^{(1)})X_{-m_r}|\Big)^2\Big\|_{P,2q}\\
&\lesssim \big\|\sup\limits_{r\in\{1,\dots,d\}}\big(\varepsilon^{(m_r)}\big)^2\big\|_{P,2q} + \big\|\sup\limits_{r\in\{1,\dots,d\}}X_k^2\big\|_{P,2q}\\
&\quad\quad+s\tau_n^2 \big\|\sup\limits_{r\in\{1,\dots,d\}}\|X_{-m_r}\|^2_{\infty}\big\|_{P,2q}\\
&\lesssim \log(d)+\log(d)+s\tau_n^2\log(a_n)\\
&\lesssim \log(a_n)
\end{align*}
and with an analogous argument 
$$ \big\|\big(F_{1,1}^{(2)}\big)^2\big\|_{P,2q}\lesssim \log(a_n).$$
Since we excluded the true nuisance parameter, which does not need to be sparse, we have $\mathcal{F}_{1,1}^{(1)}\subseteq \mathcal{G}_{1,1}$ and $\mathcal{F}_{1,1}^{(2)}\subseteq \mathcal{G}_{1,1}$
with
$$\mathcal{G}_{1,1}:=\Big\{X\to \xi X:\xi\in\R^p,\|\xi\|_0\le Cs,\|\xi\|_2\le C\Big\}$$
where $\mathcal{G}_{1,1}$ is a union over ${p}\choose{Cs}$ VC-subgraph classes $\mathcal{G}_{1,1,k}$ with VC indices less or equal to $Cs+2$ (Lemma 2.6.15, Van der Vaart and Wellner (1996)\cite{vanweak}). This implies that $\mathcal{F}_{1,1}^{(1)}$ and $\mathcal{F}_{1,1}^{(2)}$ are unions over ${p}\choose{Cs}$ VC-subgraph classes $\mathcal{F}_{1,1,k}^{(1)}$ and $\mathcal{F}_{1,1,k}^{(2)}$ with VC indices less or equal to $Cs+2$.\\
Due to theorem 2.6.7 in \cite{vanweak} we obtain
\begin{align*}
&\quad\sup_{Q}\log N(\varepsilon\|F_{1,1}^{(1)}\|_{Q,2},\mathcal{F}_{1,1}^{(1)},\|\cdot\|_{Q,2})\\
&\le\sup_{Q}\log\Bigg(\sum\limits_{k=1}^{\binom{p}{Cs}}N(\varepsilon\|F_{1,1}^{(1)}\|_{Q,2},\mathcal{F}_{1,1,k}^{(1)},\|\cdot\|_{Q,2})\Bigg)\\
&\le\log\Bigg( \underbrace{\binom{p}{Cs}}_{\le \big(\frac{e\cdot p}{Cs}\big)^{Cs}} K(Cs+2)(16e)^{Cs+2}\left(\frac{1}{\varepsilon}\right)^{2Cs+2}\Bigg)\\
&\le\log\Bigg( \left(\frac{e\cdot p}{Cs}\right)^{Cs} K(Cs+2)(16e)^{Cs+2}\left(\frac{1}{\varepsilon}\right)^{2Cs+2}\Bigg)\\
&\lesssim s\log\Big(\frac{a_n}{\varepsilon}\Big)
\end{align*}
where $K$ is an universal constant and with an analogous argument 
$$\quad\sup_{Q}\log N(\varepsilon\|F_{1,1}^{(2)}\|_{Q,2},\mathcal{F}_{1,1}^{(2)},\|\cdot\|_{Q,2})\lesssim s\log\Big(\frac{a_n}{\varepsilon}\Big).$$
Using basic calculations on covering entropies (see for example Appendix N Lemma N.1 from Belloni et al. (2014) \cite{belloni2014uniform}) we can bound the covering entropy of the class $\mathcal{F}_{1,1}$ by
\begin{align*}
&\sup_{Q}\log N(\varepsilon\|F_{1,1}^{(1)}F_{1,1}^{(2)}\|_{Q,2},\mathcal{F}_{1,1},\|\cdot\|_{Q,2})\\
\le\quad&\sup_{Q}\log N\Big(\frac{\varepsilon}{2}\|F_{1,1}^{(1)}\|_{Q,2},\mathcal{F}_{1,1}^{(1)},\|\cdot\|_{Q,2}\Big)\\
&\quad+\sup_{Q}\log N\Big(\frac{\varepsilon}{2}\|F_{1,1}^{(2)}\|_{Q,2},\mathcal{F}_{1,1}^{(2)},\|\cdot\|_{Q,2}\Big)\\
\lesssim\quad & s\log\Big(\frac{a_n}{\varepsilon}\Big)
\end{align*}
where $F_{1,1}:=F_{1,1}^{(1)}F_{1,1}^{(2)}$ is an envelope for $\mathcal{F}_{1,1}$ with 
$$\|F_{1,1}\|_{P,q}\le \Big(\big\|\big(F_{1,1}^{(1)}\big)^2\big\|_{P,2q}\big\|\big(F_{1,1}^{(1)}\big)^2\big\|_{P,2q}\Big)^{1/2}\lesssim \log(a_n).$$
Additionally define
\begin{align*}
\mathcal{F}_{1,2}:&=\Big\{\psi_{m_r}(\cdot,\theta,\eta_{m_r}):r\in\{1,\dots,d\},\theta\in\Theta_{m_r}\Big\}.
\end{align*}
With the same argument as above $\mathcal{F}_{1,2}$ is a union over $d$ VC-subgraph classes with VC indices less or equal to $3$ implying
\begin{align*}
\sup_{Q}\log N(\varepsilon\|F_{1,2}\|_{Q,2},\mathcal{F}_{1,2},\|\cdot\|_{Q,2})\le C\log\Big(\frac{d}{\varepsilon}\Big)\lesssim \log\Big(\frac{a_n}{\varepsilon}\Big)
\end{align*}
where the envelope $F_{1,2}$ of $\mathcal{F}_{1,2}$ obeys
\begin{align*}
\|F_{1,2}\|_{P,q}\lesssim \log(a_n)
\end{align*}
with an analogous argument as above. Combining these results we obtain 
\begin{align*}
&\sup_{Q}\log N(\varepsilon\|F_{1}\|_{Q,2},\mathcal{F}_{1},\|\cdot\|_{Q,2})\\
=&\quad \sup_{Q}\log N(\varepsilon\|F_{1,1}^{(1)}F_{1,1}^{(2)}\vee F_{1,2}\|_{Q,2},\mathcal{F}_{1,1}\cup\mathcal{F}_{1,2},\|\cdot\|_{Q,2})\\
\le &\quad \sup_{Q}\log N(\varepsilon\|F_{1,1}^{(1)}F_{1,1}^{(2)}\|_{Q,2},\mathcal{F}_{1,1},\|\cdot\|_{Q,2})\\
&\quad\quad+\sup_{Q}\log N(\varepsilon\|F_{1,2}\|_{Q,2},\mathcal{F}_{1,2},\|\cdot\|_{Q,2})\\
\lesssim&\quad s\log\Big(\frac{a_n}{\varepsilon}\Big)
\end{align*}
where the envelope $F_{1}:= F_{1,1}^{(1)}F_{1,1}^{(2)}\vee F_{1,2}$ of $\mathcal{F}_{1}$ satisfies 
$$\|F_{1}\|_{P,q}\lesssim \log(a_n)$$
which gives us Assumption 2.2 (iv). Observe that for all $f\in \mathcal{F}_1$ we have
\begin{align*}
\E[f^2]^{1/2}&\le\sup\limits_{r,\theta,\eta^{(1)}}\E\left[(X_j-\theta X_k-\eta^{(1)}X_{-m_r})^4\right]^{1/4}\sup\limits_{r,\eta^{(2)}}\E\left[(X_k-\eta^{(2)}X_{-m_r})^4\right]^{1/4}\\
&\lesssim \sup\limits_{\|\xi\|_2=1}\E\left[(\xi X)^4\right]^{1/2}\lesssim C
\end{align*}
and 
$$\E[f^2]^{1/2}=\E\Big[(\underbrace{X_j-\theta X_k-\eta^{(1)}X_{-m_r}}_{=:Z_1})^2(\underbrace{X_k-\eta^{(2)}X_{-m_r}}_{=:Z_2})^2\Big]^{1/2}.$$
For each $Z_i$ with $i\in\{1,2\}$ we have
$$E[Z_i^2]\gtrsim \inf\limits_{\|\xi\|_2=1}\E\left[(\xi X)^2\right]\ge c.$$
Therefore $Z_1$ and $Z_2$ are both centered normal distributed random variables where the variance is bounded away from zero. This implies
$$E[Z_1^2Z_2^2]^{1/2}\ge c>0$$
which gives us Assumption 2.2 (v).\\
Assumption 2.2 (vi) (a) holds by construction of $\tau_n$ and $v_n\lesssim s$. Due to the growth condition \ref{A2} we can choose $q=2\tilde{q}/(1-\kappa)$ such that
\begin{align*}
n^{-1/2+1/q}s\log^{2}(a_n)&=n^{\frac{1-\kappa}{2\tilde{q}}} n^{-1/2}s\log^{2}(a_n)\\
&=n^{-\frac{\kappa}{2\tilde{q}}}\left(n^{\frac{1}{\tilde{q}}} \frac{s^2\log^{4}(a_n)}{n}\right)^{1/2}\lesssim n^{-\frac{\kappa}{2\tilde{q}}}.
\end{align*}
Additionally we have
\begin{align*}
C\tau_n(s\log(a_n))^{1/2}\lesssim\frac{s\log(a_n)}{\sqrt{n}}\lesssim n^{-\frac{1}{2\tilde{q}}},
\end{align*}
\begin{align*}
\log^{1/2}(d)\frac{\log(n)}{\sqrt{n}}(s\log(a_n))^{1/2}\lesssim\sqrt{\frac{s\log^4(a_n)}{n}}\lesssim n^{-\frac{1}{2\tilde{q}}}
\end{align*}
and
$$n^{1/2}\tau_n^2=\frac{s\log(a_n)}{\sqrt{n}} \lesssim n^{-\frac{1}{2\tilde{q}}}$$ 
which gives us Assumption 2.2 (vi) (b) and (c) with $\delta_n=n^{-\frac{\kappa}{2\tilde{q}}}$.
Define the class
$$\mathcal{F}_0:=\{\bar{\psi}_{m_r}(\cdot):r=1,\dots,d\}$$
where $\bar{\psi}_{m_r}(\cdot):=-\sigma_{m_r}^{-1}J_{m_r}^{-1}\psi_{m_r}(\cdot,\theta_{m_r},\eta_{m_r})$ with
$\sigma_{m_r}^2:=J_{m_r}^{-2}\E[\psi_{m_r}^2(X,\theta_{m_r},\eta_{m_r})]$.
Observe that by the Cauchy-Schwarz Inequality for any $q>0$ the envelope $F_0$ for $\mathcal{F}_0$ satisfies
\begin{align*}
\|F_0\|_{P,q}&=\E\left[\sup\limits_{r=1,\dots,d}\Big(\E[(\varepsilon^{(m_r)}\nu^{(m_r)})^2]^{-1/2}|\varepsilon^{(m_r)}\nu^{(m_r)}|\Big)^q\right]^{1/q}\\
&\lesssim \E\left[\sup\limits_{r=1,\dots,d}\big(|\varepsilon^{(m_r)}\nu^{(m_r)}|\big)^q\right]^{1/q}\\
&\lesssim \log(d).
\end{align*}
Since $|\mathcal{F}_0|=d$ we have
$$\sup\limits_Q\log N\big(\varepsilon\|F_0\|_{Q,2},\mathcal{F}_0,\|\cdot\|_{Q,2}\big)\le \log\Big(\frac{d}{\varepsilon}\Big)$$
for all $<\varepsilon\le 1$. Therefore Assumption 2.3 (i) is satisfied with $\varrho_n=1$ and $A_n=d\vee n$. Since the errors are centered normal distributed random variables with a uniformly bounded variance we have $E\big[\big(\varepsilon^{(m_r)}\big)^8\big]\lesssim C$ and $E\big[\big(\nu^{(m_r)}\big)^8\big]\lesssim C$. This implies
$\E[f^4]\le C$ for all $f\in\mathcal{F}_0$ which gives us Assumption 2.3 (ii). 
The growth condititons from corollary $2.1$ are satisfied due to Condition \ref{A2}. Observe that
$$\delta_n^2 \log(n\vee d) \lesssim n^{-\frac{\kappa}{\tilde{q}}}\log(n\vee d)=o(1),$$
$$\log^{2/7}(d) \log(n\vee d) =o(n^{1/7})$$
and we can find a $q$ such that
$$\log^{2/3}(d) \log(n\vee d) =o(n^{1/3-2/(3q)}).$$
Now, we verify Assumption 2.4.  Define
\begin{align*}
\tilde{\psi}_{m_r}(X,\eta^{(2)}):&=-X_k(X_k-\eta^{(2)}X_{-m_r})
\end{align*}
and
\begin{align*}
\tilde{m}_{m_r}(\eta^{(2)}):&= \E[\tilde{\psi}_{m_r}(X,\eta^{(2)})],
\end{align*}
where $\hat{J}_{m_r}=-\E_n[\tilde{\psi}_{m_r}(X,\hat{\eta}^{(2)})]$. It holds
\begin{align*}
|\hat{J}_{m_r}-J_{m_r}|\le |\hat{J}_{m_r}-\tilde{m}_{m_r}(\hat{\eta}^{(2)})|+|\tilde{m}_{m_r}(\hat{\eta}^{(2)})-\tilde{m}_{m_r}(\eta^{(2)}_{m_r})|
\end{align*}
with
\begin{align*}
|\tilde{m}_{m_r}(\hat{\eta}^{(2)})-\tilde{m}_{m_r}(\eta^{(2)}_{m_r})|
&=|\E[X_k(\hat{\eta}_{m_r}^{(2)}-\eta_{m_r}^{(2)})X_{-m_r}]|\\
&=||\hat{\eta}_{m_r}^{(2)}-\eta_{m_r}^{(2)}||_2\left|\E\left[X_k\left(\frac{(\hat{\eta}_{m_r}^{(2)}-\eta_{m_r}^{(2)})}{||\hat{\eta}_{m_r}^{(2)}-\eta_{m_r}^{(2)}||_2}X_{-m_r}\right)\right]\right|\\
&\lesssim ||\hat{\eta}_{m_r}^{(2)}-\eta_{m_r}^{(2)}||_2\lesssim \tau_n.
\end{align*}
Let $$\tilde{\mathcal{G}}_1:=\{X\mapsto \tilde{\psi}_{m_r}(X,\eta^{(2)}):r=1,\dots,d, \eta^{(2)}\in\mathcal{T}^*_{m_r,2}\}$$
with
\begin{align*}
\sup\limits_r|\hat{J}_{m_r}-J_{m_r}|\lesssim\sup\limits_{g\in\tilde{\mathcal{G}}_1}|\E_n[g(X)]-\E[g(X)]|+\tau_n.
\end{align*}
The class $\tilde{\mathcal{G}}_1$ has an envelope $\tilde{G}_1$ with
\begin{align*}
\E[\tilde{G}_1^q]^{1/q}&\le\E\left[\sup\limits_r\sup\limits_{\eta^{(2)}\in\mathcal{T}^*_{m_r,2}}|X_k^q(X_k-\eta^{(2)}X_{m_r})^q|\right]^{1/q}\\
&\le||\sup\limits_r X_k||_{P,2q}\E\left[\sup\limits_{r,\eta^{(2)}\in\mathcal{T}^*_{m_r,2}}(X_k-\eta^{(2)}X_{m_r})^{2q}\right]^{1/2q}\\
&\lesssim \log^{\frac{1}{2}}(d)\left(||\sup\limits_r\nu^{(m_r)}||_{P,2q}\vee\E\left[\sup\limits_{r,\eta^{(2)}\in\mathcal{T}^*_{m_r,2}}((\eta^{(2)}_{m_r}-\eta^{(2)})X_{m_r})^{2q}\right]^{1/2q}\right)\\
&\lesssim \log^{\frac{1}{2}}(d)\left(\log^{\frac{1}{2}}(d)\vee\sqrt{s}\tau_n \sup\limits_rE\left[||X_{m_r}||_\infty^{2q}\right]^{1/2q}\right)\\
&\lesssim \log(a_n).
\end{align*}
for all $q$. With similar arguments as in the verification of Assumption 2.2. (iv), we obtain
$$\sup\limits_Q\log N\big(\varepsilon\|\tilde{G}_1\|_{Q,2},\mathcal{G}_1,\|\cdot\|_{Q,2}\big)\lesssim s\log\Big(\frac{a_n}{\varepsilon}\Big).$$
Therefore, by Lemma O.2, it holds
\begin{align*}
\sup\limits_r|\hat{J}_{m_r}-J_{m_r}|&\lesssim K\left(\sqrt{\frac{s\log(a_n)}{n}}+n^{1/q}\frac{s\log^2(a_n)}{n}\right)+\tau_n\\
&=o\left(\log^{-\frac{3}{2}}(a_n)\right)
\end{align*}
with probability not less then $1-o(1)$.
Next we want to show that
$$\E_n[\psi_{m_r}^2(X,\hat{\theta}_{m_r},\hat{\eta}_{m_r})]-\E[\psi_{m_r}^2(X,\theta_{m_r},\eta_{m_r})]=o_P(\log^{-1}(a_n)).$$
By the triangle inequality we have
\begin{align*}
&|\E_n[\psi_{m_r}^2(X,\hat{\theta}_{m_r},\hat{\eta}_{m_r})]-\E[\psi_{m_r}^2(X,\theta_{m_r},\eta_{m_r})]|\\
\le\quad &|\E_n[\psi_{m_r}^2(X,\hat{\theta}_{m_r},\hat{\eta}_{m_r})]-\E[\psi_{m_r}^2(X,\hat{\theta}_{m_r},\hat{\eta}_{m_r})]|\\
\quad &+ |\E[\psi_{m_r}^2(X,\hat{\theta}_{m_r},\hat{\eta}_{m_r})-\psi_{m_r}^2(X,\theta_{m_r},\eta_{m_r})]|\\
\le\quad &|\E_n[\psi_{m_r}^2(X,\hat{\theta}_{m_r},\hat{\eta}_{m_r})]-\E[\psi_{m_r}^2(X,\hat{\theta}_{m_r},\hat{\eta}_{m_r})]|\\
\quad &+ \E[(\psi_{m_r}(X,\hat{\theta}_{m_r},\hat{\eta}_{m_r})+\psi_{m_r}(X,\theta_{m_r},\eta_{m_r}))^2]^{1/2}\\
&\quad\cdot\E[(\psi_{m_r}(X,\hat{\theta}_{m_r},\hat{\eta}_{m_r})-\psi_{m_r}(X,\theta_{m_r},\eta_{m_r}))^2]^{1/2}\\
\le\quad &|\E_n[\psi_{m_r}^2(X,\hat{\theta}_{m_r},\hat{\eta}_{m_r})]-\E[\psi_{m_r}^2(X,\hat{\theta}_{m_r},\hat{\eta}_{m_r})]|\\
\quad &+ C(|\theta_{m_r}-\hat{\theta}_{m_r}|\vee\|\eta_{m_r}-\hat{\eta}_{m_r}\|_e)
\end{align*}
due to 2.1(a) and 2.2(v).
Observe that with probability $1-o(1)$ $$\sup\limits_r|\hat{\theta}_{m_r}-\theta_{m_r}|\lesssim \tau_n=o(\log^{-1}(a_n))$$ due to Appendix B from Belloni et al. (2018) \cite{belloni2018uniformly}.
Since the class 
$$\tilde{\mathcal{G}}_2:=\Big\{\psi_{m_r}(\cdot,\theta,\eta):r\in\{1,\dots,d\},|\theta-\theta_{m_r}|\le C\tau_n,\eta\in\mathcal{T}^*_{m_r}\Big\}\subseteq \mathcal{F}_{1,1}$$
we obtain the same entropy bounds as for $\mathcal{F}_{1,1}$ implying 
$$\sup_{Q}\log N(\varepsilon\|\tilde{G}_2^2\|_{Q,2},\tilde{\mathcal{G}}_2^2,\|\cdot\|_{Q,2})\lesssim s\log\Big(\frac{a_n}{\varepsilon}\Big)$$
where $\tilde{G}_2^2$ is a measurable envelope of $\tilde{\mathcal{G}}_2^2$ with
\begin{align*}
\|\tilde{G}_2^2\|_{P,q}&\le \|\big(F_{1,1}\big)^2\big\|_{P,q}\\
&\le\Big(\big\|\big(F_{1,1}^{(1)}\big)^4\big\|_{P,q}\big\|\big(F_{1,1}^{(2)}\big)^4\big\|_{P,q}\Big)^{1/2}\\
&\lesssim \log^2(a_n)
\end{align*}
due to $\big\|\big(F_{1,1}^{(1)}\big)^4\big\|_{P,q}\lesssim\log^2(a_n)$ and $\big\|\big(F_{1,1}^{(2)}\big)^4\big\|_{P,q}\lesssim\log^2(a_n)$.
For all $g\in \tilde{\mathcal{G}}_2^2$ we have
\begin{align*}
&\sup_{g\in\tilde{\mathcal{G}}_2^2}\E[g(X)^2]^{1/2}\\ \quad&\le\sup\limits_{r,\theta,\eta^{(1)}}\E\Big[(X_j-\theta X_k-\eta^{(1)}X_{-m_r})^8\Big]^{1/4}\sup\limits_{r,\eta^{(2)}}\E\Big[(X_k-\eta^{(2)}X_{-m_r})^8\Big]^{1/4}\\
&\lesssim \sup\limits_{||\xi||_2=1}\E\Big[(\xi X)^8\Big]^{1/2}\le C.
\end{align*}
Therefore we can find a $q>4$ such that with probability $1-o(1)$
\begin{align*}
\sup\limits_{g\in\tilde{\mathcal{G}}_2^2}|\E_n[g(X)]-\E[g(X)]|&\le K\left(\sqrt{\frac{s\log(a_n)}{n}}+n^{1/q}\frac{s\log^3(a_n)}{n}\right)\\
&=o(\log^{-1}(a_n))
\end{align*}
which implies
$$\E_n[\psi_{m_r}^2(X,\hat{\theta}_{m_r},\hat{\eta}_{m_r})]-\E[\psi_{m_r}^2(X,\theta_{m_r},\eta_{m_r})]=o_P(\log^{-1}(a_n)).$$
Since $1\lesssim \sigma^2_{m_r}\lesssim 1$ due to Assumption 2.1 (iv) and 2.2 (v), we have
\begin{align*}
\left|\frac{\hat{\sigma}_{m_r}}{\sigma_{m_r}}-1\right|&\le \left|\frac{\hat{\sigma}^2_{m_r}}{\sigma^2_{m_r}}-1\right|\\
&\lesssim \Big|\hat{\sigma}^2_{m_r}-\sigma^2_{m_r}\Big|\\
&\le \Big|\hat{J}_{m_r}^{-2}-J_{m_r}^{-2}\Big|\E_n[\psi_{m_r}^2(X,\hat{\theta}_{m_r},\hat{\eta}_{m_r})]\\
&\quad + J_{m_r}^{-2}|\E_n[\psi_{m_r}^2(X,\hat{\theta}_{m_r},\hat{\eta}_{m_r})]-\E[\psi_{m_r}^2(X,\theta_{m_r},\eta_{m_r})]|\\
&\lesssim \Big|\hat{J}_{m_r}-J_{m_r}\Big|+|\E_n[\psi_{m_r}^2(X,\hat{\theta}_{m_r},\hat{\eta}_{m_r})]-\E[\psi_{m_r}^2(X,\theta_{m_r},\eta_{m_r})]|\\
&=o_P(\log^{-1}(a_n))
\end{align*}
uniformly over all $r=1,\dots,d$ which gives us Assumption $2.4$ with $\Delta_n=o(1)$ and $\varepsilon_n=o(\log^{-1}(a_n))$. Next, we show the Assumption 2.3 (iii). The entropy conditions of the class $$\hat{\mathcal{F}}_0=\{\bar{\psi}_{m_r}(\cdot)-\hat{\psi}_{m_r}(\cdot):r=1,\dots,d\}$$
holds by construction with $\bar{A}_n=d\vee n$ and $\bar{\varrho}=1$. Further it holds for all $f\in\hat{\mathcal{F}}_0$
\begin{align*}
||f||_{P_n,2}&=||\hat{\sigma}^{-1}_{m_r}\hat{J}^{-1}_{m_r}\psi_{m_r}(X,\hat{\theta}_{m_r},\hat{\eta}_{m_r})-\sigma^{-1}_{m_r}J^{-1}_{m_r}\psi_{m_r}(X,\theta_{m_r},\eta_{m_r})||_{P_n,2}\\
&\le |\hat{\sigma}^{-1}_{m_r}\hat{J}^{-1}_{m_r}-\sigma^{-1}_{m_r}J^{-1}_{m_r}|\cdot||\psi_{m_r}(X,\theta_{m_r},\eta_{m_r})||_{P_n,2}\\
&\quad+\hat{\sigma}^{-1}_{m_r}\hat{J}^{-1}_{m_r}||\psi_{m_r}(X,\hat{\theta}_{m_r},\hat{\eta}_{m_r})-\psi_{m_r}(X,\theta_{m_r},\eta_{m_r})||_{P_n,2}\\
&:=I+II
\end{align*}
To bound the first term, observe that uniformly over all $r=1,\dots,d$
\begin{align*}
&\quad|\hat{\sigma}^{-1}_{m_r}\hat{J}^{-1}_{m_r}-\sigma^{-1}_{m_r}J^{-1}_{m_r}|=o_P(\log^{-1}(a_n))
\end{align*}
since $1\lesssim J_{m_r}\lesssim 1$ and $1\lesssim\sigma_{m_r}\lesssim 1$. Define the class
$$\tilde{\mathcal{G}}_3:=\{\psi_{m_r}^2(\cdot,\theta_{m_r},\eta_{m_r}):r=1,\dots,d\}$$
with cardinality $|\tilde{\mathcal{G}}_3|=d$ and an envelope $\tilde{G}_3$ that fulfills 
\begin{align*}
||\tilde{G}_3||_{P,q}\le \E\left[\sup\limits_r \left(\varepsilon^{(m_r)}\nu^{(m_r)}\right)^{2q}\right]^{1/q}\lesssim \log^2(d).
\end{align*}
Remark that
$$\sup\limits_r||\psi_{m_r}(X,\theta_{m_r},\eta_{m_r})||_{P_n,2}\le \left(\frac{1}{\sqrt{n}}\sup\limits_{g\in\tilde{G}_3}\mathbb{G}_n(g)+\sup\limits_r\E[\psi_{m_r}^2(X,\theta_{m_r},\eta_{m_r})]\right)^{\frac{1}{2}}$$
with $\sup\limits_r\E[\psi_{m_r}^2(X,\theta_{m_r},\eta_{m_r})]\le C$ and 
\begin{align*}
\frac{1}{\sqrt{n}}\sup\limits_{g\in\tilde{G}_3}\mathbb{G}_n(g)\lesssim K\left(\sqrt{\frac{\log (a_n)}{n}}+n^{1/q}\frac{\log^3(a_n)}{n}\right)=o(1)
\end{align*}
with probability $1-o(1)$. This implies $$I=o_P\left(\log^{-1}(a_n)\right)$$ uniformly over all $r=1,\dots,d$. To bound the second term, define the class
\begin{align*}
\tilde{\mathcal{G}}_4:=\{\psi_{m_r}(\cdot,\theta,\eta)-\psi_{m_r}(\cdot,\theta_{m_r},\eta_{m_r})&:r=1,\dots,d,\\
&|\theta-\theta_{m_r}|\le C\tau_n,\eta\in\mathcal{T}_{m_r}\}
\end{align*}
for a sufficiently large constant $C>0$. Due to Assumption 2.2 (i) we have that 
$$\psi_{m_r}(X,\hat{\theta}_{m_r},\hat{\eta}_{m_r})-\psi_{m_r}(X,\theta_{m_r},\eta_{m_r})\in \tilde{\mathcal{G}}_4$$ 
with probability $1-o(1)$. Since $\tilde{\mathcal{G}}_4^2\subseteq (\mathcal{F}_1-\mathcal{F}_1)^2$ the covering numbers obey
$$\sup\limits_Q\log N\big(\varepsilon\|\tilde{G}_4^2\|_{Q,2},\tilde{\mathcal{G}}_4^2,\|\cdot\|_{Q,2}\big)\lesssim s\log\Big(\frac{a_n}{\varepsilon}\Big)$$
and the envelope 
$$\tilde{G}_4^2=\sup\limits_{r=1,\dots,d}\sup\limits_{|\theta-\theta_{m_r}|\le C\tau_n}\sup\limits_{\eta\in\mathcal{T}_{m_r}}\left(\psi_{m_r}(\cdot,\theta,\eta)-\psi_{m_r}(\cdot,\theta_{m_r},\eta_{m_r})\right)^2$$ 
satisfies 
\begin{align*}
&\quad\|\tilde{G}_4^2\|_{P,q}\\
&\lesssim \|\sup\limits_{r,\theta,\eta^{(2)}}\left((\theta_{m_r}-\theta)X_k(X_k-\eta^{(2)}X_{-m_r})\right)^2\|_{P,q}\\
&\quad + \|\sup\limits_{r,\eta^{(1)},\eta^{(2)}}\left((X_j-\theta_{m_r}X_k-\eta^{(1)}X_{-m_r})(\eta^{(2)}_{m_r}-\eta^{(2)})X_{-m_r}\right)^2\|_{P,q}\\
&\quad + \|\sup\limits_{r,\eta^{(1)}}\left((X_k-\eta^{(2)}_{m_r}X_{-m_r})(\eta^{(1)}_{m_r}-\eta^{(1)})X_{-m_r}\right)^2\|_{P,q}\\
&:=T_1+T_2+T_3
\end{align*} 
with
\begin{align*}
T_1&\lesssim \tau_n^2 \|\sup\limits_{r,\eta^{(2)}}\left(X_k(X_k-\eta^{(2)}X_{-m_r})\right)^2\|_{P,q}\\
&\lesssim\tau_n^2\|\sup\limits_{r}X_k^2\|_{P,2q}\|\sup\limits_{r,\eta^{(2)}}(X_k-\eta^{(2)}X_{-m_r})^2\|_{P,2q}\\
&\lesssim \frac{s\log(a_n)}{n} \log(d)^2=o(\log^{-1}(a_n)),
\end{align*}
\begin{align*}
T_2&\lesssim \|\sup\limits_{r,\eta^{(2)}}((\eta^{(2)}_{m_r}-\eta^{(2)})X_{-m_r})^2\|_{P,2q}\|\sup\limits_{r,\eta^{(1)}}(X_j-\theta_{m_r}X_k-\eta^{(1)}X_{-m_r})^2\|_{P,2q}\\
&\lesssim s\tau_n^2\|\sup\limits_{r}\|X_{-m_r}\|_\infty^2\|_{P,2q}\log(d)\\
&\lesssim \frac{s^2\log(a_n)}{n}\log(a_n)\log(d)=o(\log^{-1}(a_n))
\end{align*}
and
\begin{align*}
T_3&\lesssim \|\sup\limits_{r,\eta^{(1)}}((\eta^{(1)}_{m_r}-\eta^{(1)})X_{-m_r})^2\|_{P,2q}\|\sup\limits_{r}(\nu_{m_r})^2\|_{P,2q}\\
&\lesssim s\tau_n^2\|\sup\limits_{r}\|X_{-m_r}\|_\infty^2\|_{P,2q}\log(d)=o(\log^{-1}(a_n)).
\end{align*}
Since 
$$\sigma:=\left(\sup\limits_{g\in \tilde{\mathcal{G}}_4^2}\E[g^2]\right)^{1/2}\lesssim \frac{s^2\log(a_n)}{n}=o(\log^{-3}(a_n))$$
it holds 
\begin{align*}
\frac{1}{\sqrt{n}}\sup\limits_{g\in\tilde{G}_4^2}\mathbb{G}_n(g)&\lesssim K\left(\sigma\sqrt{\frac{s\log (a_n)}{n}}+n^{1/q}\|\tilde{G}_4^2\|_{P,q}\frac{s\log(a_n)}{n}\right)
\\&=o(\log^{-4}(a_n))
\end{align*}
with probability $1-o(1)$. 
Hence,
\begin{align*}
&||\psi_{m_r}(X,\hat{\theta}_{m_r},\hat{\eta}_{m_r})-\psi_{m_r}(X,\theta_{m_r},\eta_{m_r})||_{P_n,2}\\
\le\quad& \left(\frac{1}{\sqrt{n}}\sup\limits_{g\in\tilde{G}_4^2}\mathbb{G}_n(g)+\sup\limits_{g\in\tilde{G}_4^2}\E[g(X)]\right)^{\frac{1}{2}}=o(\log^{-3/2}(a_n))
\end{align*}
with probability $1-o(1)$ due to Assumption 2.1 (v) (a).\\ \\
This gives us $II=o_p\left(\log^{-1}(a_n)\right)$ with probability $1-o(1)$ implying Assumption 2.3 (iii) with $\bar{\delta_n}=o(\log^{-1}(a_n))=o(1)$. It is straightforward to see that the growth conditions of Corollary 2.2 hold.
\end{proof}


\section{Uniform nuisance function estimation}\label{uniformestimation}
\noindent Consider the following linear regression model
$$ Y_r=\sum\limits_{j=1}^p \beta_{r,j} X_{r,j}+\varepsilon_{r}=\beta_r X_r+\varepsilon_r$$ 
with centered regressors and errors $\varepsilon_r$ with $\E[\varepsilon_r]=0$ for each $r =1,\dots,d$. The true parameter obeys
$$\beta_r\in \arg\min_{\beta}\E[(Y_r-\beta X_r)^2]$$
with
$$\beta_r=\beta^{(1)}_r+\beta^{(2)}_r.$$
The parameter $\beta^{(2)}_r$ is the approximate sparse part of the true regression coefficient that captures the misspecification of a sparse model. We show that the lasso, post-lasso and square-root lasso estimators have sufficiently fast estimation rates uniformly for all $r=1,\dots,d$. In this setting $d=d_n$ is explicitly allowed to grow with $n$. In the following analysis, the regressors and errors need to have at least subexponential tails. In this context, we define the Orlicz norm $\|X\|_{\Psi_\rho}$ as
\begin{align*}
\|X\|_{\Psi_\rho}=\inf\{C>0:\E[\Psi_\rho(|X|/C)]\le 1\}
\end{align*} 
with $\Psi_\rho(x)=\exp(x^{\rho})-1$.


\subsection{Uniform lasso estimation}\label{uniformla}
Define the weighted lasso estimator
\begin{align*}
 \hat{\beta}_r\in \arg\min\limits_{\beta}\left(\frac{1}{2}\E_n\left[\left(Y_r-\beta X_r\right)^2\right]+\frac{\lambda}{n}\|\hat{\Psi}_{r,m}\beta\|_1\right)
\end{align*}
with the penalty level
\begin{align*}
\lambda = c_\lambda\sqrt{n}\Phi^{-1}\left(1-\frac{\gamma}{2pd}\right)
\end{align*}
for a suitable $c_\lambda >1$, $\gamma\in[1/n,1/\log(n)]$ and a fix $m\ge 0$. Define the post-regularized weighted least squares estimator as
\begin{align*}
\tilde{\beta}_r\in \arg\min\limits_{\beta}\left(\frac{1}{2}\E_n\left[\left(Y_r-\beta X_r\right)^2\right]\right):\quad \text{supp}(\beta)\subseteq \text{supp}(\hat{\beta}_r).
\end{align*}
The penalty loadings $\hat{\Psi}_{r,m}=\diag(\{\hat{l}_{r,j,m},j=1,\dots,p\})$ are defined by
$$\hat{l}_{r,j,0}=\max\limits_{1\le i\le n}||X_r^{(i)}||_\infty$$ 
for $m=0$ and for all $m\ge 1$ by the following algorithm:
\begin{algorithm}\label{penalg}
Set $\bar{m}=0$. Compute $\hat{\beta}_r$ based on $\hat{\Psi}_{r,\bar{m}}$. Set $\hat{l}_{r,j,\bar{m}+1}=\E_n\left[\left(\left(Y_r-\hat{\beta}_r X_r\right) X_{r,j}\right)^2\right]^{1/2}$. If $\bar{m}=m$ stop and report the current value of $\hat{\Psi}_{r,m}$, otherwise set $\bar{m}=\bar{m}+1$.
\end{algorithm}
\noindent Let $a_n:=\max(p,n,d,e)$. In order to establish uniform convergence rates, the following assumptions are required to hold uniformly in $n\ge n_0,P\in\mathcal{P}_n$:\\ \\
Assumptions \textbf{B1}-\textbf{B4}:
\begin{enumerate}[label=\textbf{B\arabic*},ref=B\arabic*]
\item\label{tails} \textbf{(Tail conditions)}\\
\begin{em}
There exists $1\le\rho\le 2$ such that 
$$\max\limits_{r=1,\dots,d}\max\limits_{j=1,\dots,p}\|X_{r,j}\|_{\Psi_\rho}\le C\ \text{and}\ \max\limits_{r=1,\dots,d}\|\varepsilon_{r}\|_{\Psi_\rho}\le C.$$
\end{em}
\item\label{eigen} \textbf{(Uniformly bounded eigenvalues)}\\
\begin{em}
For all $r=1,\dots,d_n$, it holds
\begin{align*}
\inf\limits_{\|\xi\|_2=1} \E\left[(\xi X_r)^2\right]\ge c \text{,} \sup\limits_{\|\xi\|_2=1} \E\left[(\xi X_r)^2\right]\le C
\end{align*}
and
\begin{align*}
\min\limits_{r=1,\dots,d}\min\limits_{j=1,\dots,p}\E[\varepsilon^2_rX_{r,j}^2]\ge c.
\end{align*}
\end{em}
\item\label{asparse} \textbf{(Uniform approximate sparsity)}\\
\begin{em}
The coefficients obey
$$\max\limits_{r=1,\dots,d}\|\beta^{(2)}_r\|_1^2\lesssim \sqrt{\frac{s^2\log(a_n)}{n}},\quad \max\limits_{r=1,\dots,d}\E\left[ (\beta^{(2)}_r X_r)^2 \right]\lesssim \frac{s\log(a_n)}{n}$$
and
$$ \max\limits_{r=1,\dots,d}\|\beta^{(1)}_r\|_0\le s.$$ 
\end{em}
\item\label{growthc} \textbf{(Growth conditions)}\\
\begin{em}
There exists a positive number $\tilde{q}>0$ such that the following growth condition is fulfilled:
\begin{align*}
n^{\frac{1}{\tilde{q}}}\frac{s\log^{1+\frac{4}{\rho}}(a_n)}{n}=o(1).
\end{align*}
\end{em}
\end{enumerate}
\begin{theorem}\label{uniformlasso}
Under the assumptions \ref{tails}-\ref{growthc} the lasso estimator $\hat{\beta}_r$ obeys uniformly over all $P\in \mathcal{P}_n$ with probability $1-o(1)$
\begin{align}
\max\limits_{r=1,\dots,d}\|\hat{\beta}_r-\beta^{(1)}_r\|_2&\le C\sqrt{\frac{s\log(a_n)}{n}},\\ \max\limits_{r=1,\dots,d}\|\hat{\beta}_r-\beta^{(1)}_r\|_1&\le C\sqrt{\frac{s^2\log(a_n)}{n}}
\end{align}
with
\begin{align}
\max\limits_{r=1,\dots,d}\|\hat{\beta}_r\|_0\le C s.
\end{align}
Additionally the post-lasso estimator $\tilde{\beta}_r$ obeys uniformly over all $P\in \mathcal{P}_n$ with probability $1-o(1)$
\begin{align}
\max\limits_{r=1,\dots,d}\|\tilde{\beta}_r-\beta^{(1)}_r\|_2&\le C\sqrt{\frac{s\log(a_n)}{n}},\\ \max\limits_{r=1,\dots,d}\|\tilde{\beta}_r-\beta^{(1)}_r\|_1&\le C\sqrt{\frac{s^2\log(a_n)}{n}}.
\end{align}
\end{theorem}


\subsection{Uniform square-root lasso estimation}\label{uniformsq}
Now, assume that $X_{r,j}$ are standardized covariates ($\E[X_{r,j}^2]=1$ for all $j=1,\dots,p$ and $r=1,\dots,d$) which are independent from the errors $\varepsilon_r$.
Define
$$Q_r(\beta):=\E_n[(Y_r-\beta X_r-\beta_r^{(2)}X_r)^2].$$ The square-root lasso estimator is definded as
$$\hat{\beta}_r\in \arg\min_{\beta}\left(\hat{Q}_r^{1/2}(\beta)+\frac{\lambda}{n}\| \beta\|_1\right),$$
where $\hat{Q}_r(\beta):=\E_n[(Y_r-\beta X_r)^2]$. $\hat{Q}_r(\beta)$ is a proxy for $Q_r(\beta)$ estimating the approximate sparse part $\beta^{(2)}_r$ by $\hat{\beta}^{(2)}_r=0$. Let
\begin{align}
\lambda=c'\sqrt{n}\Phi^{-1}\big(1-\gamma/(2pd)\big)\label{lambda}
\end{align}
where $1-\gamma$ is a confidence level associated with the probability of the event (\ref{event1}), and $c'>c $ is a slack constant. 
The first part of the analysis is to control the event
\begin{align}
\frac{\lambda}{n}\ge c\max_{r=1,\dots,d}\|S_r\|_{\infty},\label{event1}
\end{align}
where 
$$S_r := \partial_\beta Q^{1/2}(\beta)|_{\beta=\beta_r^{(1)}}=-\frac{\E_n[X_r(Y_u-\beta_r^{(1)} X_r-\beta_r^{(2)}X_r)]}{\sqrt{\E_n[(Y_u-\beta_r^{(1)} X_r-\beta_r^{(2)}X_r)^2]}}=-\frac{\E_n[X_r\varepsilon_r]}{\sqrt{\E_n[\varepsilon_r^2]}}$$ 
is the score of $Q^{1/2}$ at $\beta_r^{(1)}$. Define
$$\hat{S}_r := \partial_\beta \hat{Q}^{1/2}(\beta)|_{\beta=\beta_r^{(1)}}=-\frac{\E_n[X_r(\varepsilon_r+\beta_r^{(2)}X_r)]}{\sqrt{\E_n[(\varepsilon_r+\beta_r^{(2)}X_r)^2]}}.$$
The following conditions and lemma \ref{condwl} are essentially the same as condition WL and lemma L.4. in Belloni et al. (2018) \cite{belloni2018uniformly}.
Let $\underline{C}$ and $\overline{C}$ be some strictly positive constants. Additionally let $(\varphi_n)_{n\ge 1}$, $(\tilde{\varphi}_n)_{n\ge 1}$, $(\bar{\varphi}_n)_{n\ge 1}$ and $\Delta_n$ be some sequences of positive constants converging to zero.\\ \\
\textbf{Condition} \setword{\textbf{WL}}{WL}
\begin{em}
The following conditions hold:
\begin{enumerate}
\item $\max_{r=1,\dots,d}\max_{j=1,\dots,p}\left(\E\left[|X_{r,j}\varepsilon_r|^3\right]\right)^{1/3}\Phi^{-1}(1-\gamma/(2pd))\le \varphi_n n^{1/6}$;\label{WL1}
\item $\underline{C}\le \E\left[|X_{r,j}\varepsilon_r|^2\right]\le \overline{C}$, for all $r=1,\dots,d$ and $j=1,\dots,p$;\label{WL2}
\item with probability at least $1-\frac{1}{2}\Delta_n$,\label{WL3}
\begin{align*}
\max_{r=1,\dots,d}\max_{j=1,\dots,p} |\E_n[X_{r,j}^2\varepsilon_r^2]-\E[X_{r,j}^2\varepsilon_r^2]|\le \tilde{\varphi}_n
\end{align*}
and 
\begin{align*}
\max_{r=1,\dots,d} |\E_n[\varepsilon_r^2]-\E[\varepsilon_r^2]|\le \bar{\varphi}_n.
\end{align*}
\end{enumerate}
\end{em}
The following lemma proves that $\lambda$ satisfies (\ref{event1}) with high probability.
\begin{lemma}\label{condwl}
Suppose that condition \ref{WL} holds. In addition suppose that $\lambda$ satisfies (\ref{lambda}) for some $c'>c$ and $\gamma=\gamma_n\in[1/n,1/\log(n)]$. Then it holds
$$P\left(\frac{\lambda}{n}\ge c\max_{r=1,\dots,d}\|S_r\|_{\infty}\right)\ge 1-\gamma-o(\gamma)-\Delta_n.$$
\end{lemma}
Under the same uniform sparsity and regularity conditions as in theorem \ref{uniformlasso} we are able to show that condition \ref{WL} is satisfied and hence we can establish uniform convergence rates of the square-root lasso estimator. In section \ref{uniformsq} we additionally assumed independence between the regressors and the error terms. This eliminates the need to estimate the penalty loadings.
\begin{theorem}\label{uniform1}
Suppose that the conditions \ref{tails}-\ref{growthc} hold. In addition suppose that $\lambda$ satisfies (\ref{lambda}) for some $c'>c$ and $\gamma=\gamma_n\in[1/n,1/\log(n)]$. Then, with probability at least $1-o(1)$ we have
\begin{align}
\max\limits_{r=1,\dots,d}\|\hat{\beta}_r-\beta^{(1)}_r\|_2&\le C\sqrt{\frac{s\log(a_n)}{n}},\\ \max\limits_{r=1,\dots,d}\|\hat{\beta}_r-\beta^{(1)}_r\|_1&\le C\sqrt{\frac{s^2\log(a_n)}{n}}
\end{align}
with
\begin{align}
\max\limits_{r=1,\dots,d}\|\hat{\beta}_r\|_0\le Cs.
\end{align}
\end{theorem}
\ \\
\subsection{Proofs}\label{uniformproofs}
\begin{proof}[Proof of Theorem \ref{uniformlasso}]\ \\
Due to condition \ref{tails} we can bound the $q$-th moments of the maxima of the regressors uniformly by  
\begin{align*}
\E\left[\max\limits_{r=1,\dots,d}\|X_r\|_\infty^q\right]^{\frac{1}{q}}&=\|\max\limits_{r=1,\dots,d}\max\limits_{j=1,\dots,p}|X_{r,j}|\|_{P,q}\\
&\le q!\| \max\limits_{r=1,\dots,d}\max\limits_{j=1,\dots,p}|X_{r,j}|\|_{\psi_1}\\
&\le q!\log^{\frac{1}{\rho}-1}(2) \| \max\limits_{r=1,\dots,d}\max\limits_{j=1,\dots,p}|X_{r,j}|\|_{\psi_\rho}\\
&\le q!\log^{\frac{1}{\rho}-1}(2)K\log^{\frac{1}{\rho}}(1+dp) \max\limits_{r=1,\dots,d}\max\limits_{j=1,\dots,p}\| X_{r,j}\|_{\psi_\rho}\\
&\le C\log^{\frac{1}{\rho}}(a_n)
\end{align*}
where $C$ does depend on $q$ and $\rho$ but not on $n$. For the norm inequalities we refer to van der Vaar and Wellner (1996) \cite{vanweak}.\\
As in the previous proof we use $C$ for a strictly positive constant, independent of $n$, which may have a different value in each appearance. The notation $a_n\lesssim b_n$ stands for $a_n\le Cb_n$ for all $n$ for some fixed $C$. Additionally $a_n=o(1)$ stands for uniform convergence towards zero meaning there exists sequence $(b_n)_{n\ge 1}$ with $|a_n|\le b_n$, $b_n$ is independent of $P\in\mathcal{P}_n$ for all $n$ and $b_n\to 0$. Finally, the notation $a_n\lesssim_{P}b_n$ means that for any $\epsilon>0$, there exists $C$ such that uniformly over all $n$ we have $P_P(a_n>Cb_n)\le \epsilon$.\\
We essentially modify the proof from theorem $4.2$ from Belloni et al. (2018) \cite{belloni2018uniformly} to fit our setting and keep the notation as similar as possible.\\
We set $\mathcal{U}=\{1,\dots,d\}$ and
$$\beta_r^{(1)}\in\arg\min\limits_{\beta\in\mathbb{R}^p} \E\Big[ \underbrace{\frac{1}{2} \left(Y_r-\beta X_r-\beta^{(2)}_r X_r\right)^2}_{:=M_r(Y_r,X_r,\beta,a_r)}\Big]$$
with $a_r=\beta^{(2)}_r X_r$ for all $r=1,\dots,d$. Since the coefficient $\beta^{(2)}$ is approximately sparse by assumption we estimate the nuisance parameter $a_r$ with $\hat{a}_r\equiv 0$. Define
$$ M_r(Y_r,X_r,\beta):=M_r(Y_r,X_r,\beta,\hat{a}_r)=\frac{1}{2} \left(Y_r-\beta X_r\right)^2.$$ Then we have
$$\hat{\beta}_r\in \arg\min\limits_{\beta\in\mathbb{R}^p}\left(\E_n\left[M_r(Y_r,X_r,\beta)\right]+\frac{\lambda}{n}\|\hat{\Psi}_r\beta\|_1\right)$$
and 
$$\tilde{\beta}_r\in \arg\min\limits_{\beta\in\mathbb{R}^p}\left(\E_n\left[M_r(Y_r,X_r,\beta)\right]\right):\quad \text{supp}(\beta)\subseteq \text{supp}(\hat{\beta}_r).$$
At first we verify the condition WL from Belloni et al. (2018) \cite{belloni2018uniformly}.\\
Since $N_n=d$ we have $N(\varepsilon,\mathcal{U},d_{\mathcal{U}})\le N_n$ for all $\varepsilon\in(0,1)$ with 
$$d_{\mathcal{U}}(i,j)=\begin{cases}0 \quad \text{for }i=j\\1\quad \text{for }i\neq j.\end{cases}$$
To prove WL(i) observe that
\begin{align*}
S_r=\partial_\beta M_r(Y_r,X_r,\beta,a_r)|_{\beta=\beta^{(1)}_r}=-\varepsilon_r X_r.
\end{align*}
Since $\Phi^{-1}(1-t)\lesssim\sqrt{\log(1/t)}$, uniformly over $t\in (0,1/2)$ we have that
\begin{align*}
\|S_{r,j}\|_{P,3}\Phi^{-1}(1-\gamma/2pd)&=\|\varepsilon_r X_{r,j}\|_{P,3}\Phi^{-1}(1-\gamma/2pd)\\
&\le \left(\|\varepsilon_r\|_{P,6} \|X_{r,j}\|_{P,6}\right)^{1/2}\Phi^{-1}(1-\gamma/2pd)\\
&\le C\log^{\frac{1}{2}}(a_n)\lesssim \varphi_n n^{\frac{1}{6}}=o(1)
\end{align*}
with $$\varphi_n=O\left(\frac{\log^{\frac{1}{2}}(a_n)}{n^{\frac{1}{6}}}\right)$$ uniformly over all $j=1,\dots,p$ and $r=1,\dots,d$ by assumption \ref{tails} and \ref{growthc}.
Further, it holds
\begin{align*}
c\le\E\left[S_{r,j}^2\right]&=\E\left[\varepsilon_r^2 X_{r,j}^2\right]\\
&\le \left(\E\left[\varepsilon_r^4\right]\E\left[ X_{r,j}^4\right]\right)^{1/2}\\
&\le C
\end{align*}
for all $j=1,\dots,p$ and $r=1,\dots,d$ by assumption \ref{tails} and \ref{eigen} which implies condition $WL(ii)$. Observe that condition $WL(iii)$ reduces to 
\begin{align*}
\max\limits_{r=1,\dots,d}\max\limits_{j=1,\dots,p}|(\E_n-\E)[S_{r,j}^2]|\le \varphi_n
\end{align*}
with probability $1-\Delta_n$. We use a maximal inequality, see for example lemma $O.2$ from Belloni et al. (2018) \cite{belloni2018uniformly}. Let $\mathcal{W}=(\mathcal{Y},\mathcal{X})$ with $Y=(Y_1,\dots,Y_d)\in\mathcal{Y}$ and $X=(X_1,\dots,X_d)\in\mathcal{X}$.
Define
\begin{align*}
\mathcal{F}:=\{f_{r,j}^2|r=1,\dots,d, j=1,\dots,p\}
\end{align*}
with
\begin{align*}
f_{r,j}:&\mathcal{W}=(\mathcal{Y},\mathcal{X})\to \mathbb{R}\\
& W=(Y,X) \mapsto \left(Y_r-\beta_rX_r\right)X_{r,j}=\varepsilon_rX_{r,j}=S_{r,j}.
\end{align*}
Observe that
\begin{align*}
\|\sup\limits_{f\in\mathcal{F}}|f|\|_{P,q}&=\|\max\limits_{r=1,\dots,d}\max\limits_{j=1,\dots,p}|f_{r,j}^2|\|_{P,q}\\
&=\E\left[\max\limits_{r=1,\dots,d}\max\limits_{j=1,\dots,p}\varepsilon_r^{2q} X_{r,j}^{2q}\right]^{1/q}\\
&\le\E\left[\max\limits_{r=1,\dots,d}\varepsilon_r^{2q} \max\limits_{r=1,\dots,d}\max\limits_{j=1,\dots,p}X_{r,j}^{2q}\right]^{1/q}\\
&\le \left(\E\left[\max\limits_{r=1,\dots,d}\varepsilon_r^{4q}\right]^{1/4q} \E\left[\max\limits_{r=1,\dots,d}\max\limits_{j=1,\dots,p}X_{r,j}^{4q}\right]^{1/4q}\right)^2\\
&\le C\log^{\frac{4}{\rho}}(a_n).
\end{align*}
Since we have 
\begin{align*}
\sup\limits_{f\in\mathcal{F}}\|f\|_{P,2}^2&=\max\limits_{r=1,\dots,d}\max\limits_{j=1,\dots,p} \E\left[S_{r,j}^4\right]\le \max\limits_{r=1,\dots,d}\max\limits_{j=1,\dots,p} \E\left[\varepsilon_r^{8}\right]^{1/2}\E\left[X_{r,j}^{8}\right]^{1/2}\le C
\end{align*}
we can choose a constant with
$$\sup\limits_{f\in\mathcal{F}}\|f\|_{P,2}^2\le C\le \|\sup\limits_{f\in\mathcal{F}}|f|\|^2_{P,2}.$$
Additionally $|\mathcal{F}|=dp$ which implies
$$\log\sup\limits_{Q}N(\epsilon\|F\|_{Q,2},\mathcal{F},\|\cdot\|_{Q,2})\le  \log(dp)\lesssim \log(a_n/\epsilon), \quad 0<\epsilon\le 1.$$
Using lemma $O.2$ from Belloni et al. (2018) \cite{belloni2018uniformly} we obtain with probability not less than $1-o(1)$
\begin{align*}
\max\limits_{r=1,\dots,d}\max\limits_{j=1,\dots,p}|(\E_n-\E)[S_{r,j}^2]|&=n^{-1/2}\sup\limits_{f\in\mathcal{F}}|\mathbb{G}_n(f)|\\
&\le n^{-1/2}C\left(\sqrt{\log\left(a_n\right)}+n^{-1/2+1/q}\log^{1+\frac{4}{\rho}}(a_n)\right)\\
&=C\left(\sqrt{\frac{\log\left(a_n\right)}{n}}+\frac{\log^{1+\frac{4}{\rho}}(a_n)}{n^{1-1/q}}\right)\\
&\le \varphi_n=o(1)
\end{align*}
by the growth condition \ref{growthc}.\\
We proceed by verifying assumption $L.1$. The function $\beta\mapsto M_r\left(Y_r,X_r,\beta\right)$ is convex, which is the first requirement of assumption $L.1$.\\
We now proceed with a simplified version of proof of $J.1$ from Belloni et al. (2018) \cite{belloni2018uniformly}.
Define 
\begin{align*}
\mathcal{G}:=\{g_r:X\to(\beta^{(2)}_rX_r)^2|r=1,\dots,d\}
\end{align*}
with envelope
\begin{align*}
G:=\max\limits_{r=1,\dots,d}\|X_r\|^2_\infty\|\beta_r^{(2)}\|_1^2.
\end{align*}
Note that
\begin{align*}
\|G\|_{P,q}&=\E\left[\max\limits_{r=1,\dots,d}\|X_r\|^{2q}_\infty\|\beta_r^{(2)}\|_1^{2q}\right]^{\frac{1}{q}}\\
&\le\max\limits_{r=1,\dots,d}\|\beta_r^{(2)}\|_1^{2}\E\left[\max\limits_{r=1,\dots,d}\|X_r\|^{2q}_\infty\right]^{\frac{1}{q}}\\
&\lesssim \max\limits_{r=1,\dots,d}\|\beta_r^{(2)}\|_1^{2} \log(a_n)^{\frac{2}{\rho}}
\end{align*}
and for all $0<\varepsilon\le 1$ we have
\begin{align*}
N(\varepsilon\|G\|_{P,2},\mathcal{G},\|\cdot\|_{P,2})\le d\le d/\varepsilon.
\end{align*}
Since
\begin{align*}
\sup\limits_{g\in\mathcal{G}}\|g\|^2_{P,2}=\max\limits_{r=1,\dots,d}\E[(\beta^{(2)}_rX_r)^4]\lesssim \max\limits_{r=1,\dots,d}\|\beta_r^{(2)}\|_1^{4}
\end{align*}
we can use lemma $O.2$ from Belloni et al. (2018) \cite{belloni2018uniformly} to obtain with probability not less than $1-o(1)$
\begin{align*}
&\quad\max\limits_{r=1,\dots,d}|(\E_n-\E)[(\beta^{(2)}_rX_r)^2]|\\
&=n^{-1/2}\sup\limits_{g\in\mathcal{G}}|\mathbb{G}_n(g)|\\
&\lesssim C\left(\sqrt{\frac{\log(a_n)\max\limits_{r=1,\dots,d}\|\beta_r^{(2)}\|_1^{4}}{n}}+n^{-1+1/q}\max\limits_{r=1,\dots,d}\|\beta_r^{(2)}\|_1^{2} \log^{1+\frac{2}{\rho}}(a_n)\right)\\
&\lesssim C\left(\sqrt{\frac{\log(a_n)}{n}}\sqrt{\frac{s^2\log(a_n)}{n}}+\frac{s\log(a_n)}{n}\sqrt{n^{2/q}\frac{\log^{1+\frac{4}{\rho}}(a_n)}{n}}\right)\\
&\lesssim \frac{s\log(a_n)}{n},
\end{align*}
for a suitable choice of $q$ where we used $\max_{r=1,\dots,d}\|\beta^{(2)}_r\|_1^2\lesssim \sqrt{\frac{s^2\log(a_n)}{n}}$ from condition \ref{asparse} and the growth condition \ref{growthc}.\\
Using the triangle inequality and $\max_{r=1,\dots,d}\E\left[ (\beta^{(2)}_r X_r)^2 \right]\lesssim \frac{s\log(a_n)}{n}$ from condition \ref{asparse} we obtain
\begin{align}\label{rateapprox}
\max\limits_{r=1,\dots,d}\E_n[(\beta^{(2)}_rX_r)^2]&\le\max\limits_{r=1,\dots,d}|(\E_n-\E)[(\beta^{(2)}_rX_r)^2]|+\max\limits_{r=1,\dots,d}\E[(\beta^{(2)}_rX_r)^2]\notag \\
&\lesssim_P \frac{s\log(a_n)}{n}.
\end{align}
To show assumption $L.1$ (a), note that for all $\delta\in\mathbb{R}^p$
\begin{align*}
&\quad\left|\E_n\left[\partial_\beta M_r(Y_r,X_r,\beta_r^{(1)})-\partial_\beta M_r(Y_r,X_r,\beta_r^{(1)},a_r)\right]^T\delta\right|\\
&=\left|\E_n\left[X_r(\beta_r^{(2)}X_r)\right]^T\delta\right|\le||(\beta_r^{(2)}X_r)||_{\mathbb{P}_n,2}||X_r^T\delta||_{\mathbb{P}_n,2}\\
&\lesssim_P\sqrt{\frac{s\log(a_n)}{n}}||X_r^T\delta||_{\mathbb{P}_n,2}
\end{align*}
for all $r=1,\dots,d$. Further we have
\begin{align*}
&\quad\E_n\left[\frac{1}{2}\left(Y_r-(\beta_r^{(1)}+\delta^T)X_r\right)^2\right]-\E_n\left[\frac{1}{2}\left(Y_r-\beta_r^{(1)}X_r\right)^2\right]\\
&=-\E_n\left[\left(Y_r-\beta_r^{(1)}X_r\right)\delta^T X_r\right]+\frac{1}{2}\E_n\left[(\delta^T X_r)^2\right],
\end{align*}
where
$$-\E_n\left[\left(Y_r-\beta_r^{(1)}X_r\right)\delta^T X_r\right]=\E_n\left[\partial_\beta M_r(Y_r,X_r,\beta_r^{(1)})\right]^T\delta$$
and
$$\frac{1}{2}\E_n\left[(\delta^T X_r)^2\right]=||\sqrt{w_r}\delta^T X_r||_{\bP_n,2}^2$$
with $\sqrt{w_r}=1/4$. This gives us assumption $L.1$ (c) with $\Delta_n=0$ and $\bar{q}_{A_r}=\infty$.
Since condition $WL(ii)$ and $WL(iii)$ hold we have with probability $1-o(1)$ 
\begin{align*}
1\lesssim l_{r,j}=\left(\E_n[S_{r,j}^2]\right)^{1/2}\lesssim 1
\end{align*}
uniformly over all $r=1,\dots,d$ and $j=1,\dots,p$, which directly implies
\begin{align*}
1\lesssim\|\hat{\Psi}^{(0)}_r\|_{\infty}:=\max\limits_{j=1,\dots,p}|l_{r,j}|\lesssim 1
\end{align*}
and additionally 
\begin{align*}
1\lesssim\|(\hat{\Psi}^{(0)}_r)^{-1}\|_{\infty}:=\max\limits_{j=1,\dots,p}|l_{r,j}^{-1}|\lesssim 1.
\end{align*}
For now, we suppose that $m=0$ in algorithm \ref{penalg}. Uniformly over $r=1,\dots,d$, $j=1,\dots,p$ we have 
\begin{align*}
\hat{l}_{r,j,0}&=\left(\E_n[\max\limits_{1\le i \le n} \|X_r^{(i)}\|^2_{\infty}]\right)^{1/2}\ge \left(\E_n[ \|X_r\|^2_{\infty}]\right)^{1/2}\gtrsim_{P} 1
\end{align*}
where the last inequality holds due to condition \ref{eigen} and an application of the maximal inequality lemma.\\
Also uniformly over $r=1,\dots,d$, $j=1,\dots,p$ we have for an arbitrary $q>0$
\begin{align*}
\hat{l}_{r,j,0}&=\max\limits_{1\le i \le n} \|X_r^{(i)}\|_{\infty}\\
&\le n^{1/q}\left(\frac{1}{n}\sum\limits_{i=1}^n \|X_r^{(i)}\|_{\infty}^{q}\right)^{1/q}\\
&= n^{1/q}\left(\E_n[\|X_r\|_{\infty}^{q}]\right)^{1/q}
\end{align*}
with 
$$\E[\|X_r\|_{\infty}^{q}]^{1/q}\lesssim \log^{\frac{1}{\rho}}(a_n)$$
By maximal inequality, we have with probability $1-o(1)$ for a sufficiently large $q'>0$
\begin{align*}
&\quad\max\limits_r\left|\E_n[\|X_r\|_{\infty}^{q}]-\E[\|X_r\|_{\infty}^{q}]\right|\\
&\lesssim C\left(\sqrt{\frac{\log^{\frac{2q}{\rho}+1}(a_n)}{n}}+n^{1/q'-1}\log^{\frac{q}{\rho}+1}(a_n)\right)\\
&\lesssim \log^{\frac{q}{\rho}}(a_n)
\end{align*}
since
$$\E[\max\limits_r\|X_r\|_{\infty}^{qq'}]^{1/q'}\lesssim \log^{\frac{q}{\rho}}(a_n) \text{ and } \max\limits_r\E[\|X_r\|_{\infty}^{q2}]^{1/2}\lesssim \log^{\frac{q}{\rho}}(a_n).$$
We conclude
\begin{align*}
&\hat{l}_{r,j,0}\le n^{1/q}\left(\E_n[\|X_r\|_{\infty}^{q}]\right)^{1/q}\\
&\le n^{1/q}\left(\left|\E_n[\|X_r\|_{\infty}^{q}]-\E[\|X_r\|_{\infty}^{q}]\right|+ \E[\|X_r\|_{\infty}^{q}]\right)^{1/q}\\
&\lesssim_P n^{1/q}\log^{\frac{1}{\rho}}(a_n).
\end{align*}
uniformly over $r$. Therefore assumption $L.1(b)$ holds for some $\Delta_n=o(1)$, $L\lesssim n^{1/q}\log^{\frac{1}{\rho}}(a_n)$ and $l\gtrsim 1$. 
Hence, we can find a $c_l$ with $l>1/c_l$. Setting $c_\lambda>c_l$ and $\gamma=\gamma_n\in[1/n,1/\log(n)]$ in the choice of $\lambda$, we have 
\begin{align*}
P\left(\frac{\lambda}{n}\ge c_l\max\limits_{r=1,\dots,d}\|(\hat{\Psi}_r^{(0)})^{-1}\E_n[S_r]\|_{\infty}\right)\ge 1-\gamma-o(\gamma)-\Delta_n=1-o(1)
\end{align*}
due to lemma $L.4$ from Belloni et al. (2018) \cite{belloni2018uniformly}.\\ 
Now we uniformly bound the sparse eigenvalues. Set 
$$l_n=\log^{\frac{2}{\rho}}(a_n)n^{2/\bar{q}}$$
for a $\bar{q} > 5\tilde{q}$ with $\tilde{q}$ in \ref{growthc}. We apply Lemma $P.1$ in \cite{belloni2018uniformly} with $K\lesssim n^{1/\bar{q}}\log^{\frac{1}{\rho}}(a_n)$ and
\begin{align*}
\delta_n&\lesssim K\sqrt{sl_n}n^{-1/2}\log(sl_n)\log^{\frac{1}{2}}(a_n)\log^{\frac{1}{2}}(n)\\
&\lesssim \sqrt{n^{\frac{4}{\bar{q}}}\log(n)\log^2(sl_n)\frac{s\log^{1+\frac{4}{\rho}}(a_n)}{n}}\\
&\lesssim \sqrt{n^{\frac{5}{\bar{q}}}\frac{s\log^{1+\frac{4}{\rho}}(a_n)}{n}}
\end{align*}
for $n$ large enough. Hence by growth condition \ref{growthc}, it holds 
$$\delta_n=o(1)$$
which implies
\begin{align*}
1\lesssim \min\limits_{\|\delta\|_0\le l_n s}\frac{\|\delta X_r\|^2_{\bP_n,2}}{\|\delta\|_2^2}\le \max\limits_{\|\delta\|_0\le l_n s}\frac{\|\delta X_r\|^2_{\bP_n,2}}{\|\delta\|_2^2}\lesssim 1
\end{align*}
with probability $1-o(1)$ uniformly over $r=1,\dots,d$.\\
Define $T_r:=\text{supp}(\beta^{(1)}_r)$ and
$$\tilde{c}:=\frac{Lc_l+1}{lc_l-1}\max\limits_{r=1,\dots,d}\|\hat{\Psi}^{(0)}_r\|_{\infty}\|(\hat{\Psi}^{(0)}_r)^{-1}\|_{\infty}\lesssim L.$$
Let the restricted eigenvalues be definied as 
\begin{align*}
\bar{\kappa}_{2\tilde{c}}:=\min\limits_{r=1,\dots,d}\inf\limits_{\delta\in\Delta_{2\tilde{c},r}}\frac{\|\delta X_r\|_{\bP_n,2}}{\|\delta_{T_r}\|_2} 
\end{align*}
where $\Delta_{2\tilde{c},r}:=\{\delta:\|\delta^c_{T_r}\|_1\le 2\tilde{c}\|\delta_{T_r}\|_1\}$.
By the argument given in Bickel et al. (2009) \cite{BickelRitovTsybakov2009} we have
\begin{align*}
\bar{\kappa}_{2\tilde{c}}&\ge \left(\min\limits_{\|\delta\|_0\le l_n s}\frac{\|\delta X_r\|^2_{\bP_n,2}}{\|\delta\|_2^2}\right)^{1/2}-2\tilde{c}\left(\max\limits_{\|\delta\|_0\le l_n s}\frac{\|\delta X_r\|^2_{\bP_n,2}}{\|\delta\|_2^2}\right)^{1/2}\left(\frac{s}{sl_n}\right)^{1/2}\\
&\gtrsim \left(\min\limits_{\|\delta\|_0\le l_n s}\frac{\|\delta X_r\|^2_{\bP_n,2}}{\|\delta\|_2^2}\right)^{1/2}-2n^{\frac{1}{q}-\frac{1}{\bar{q}}}\left(\max\limits_{\|\delta\|_0\le l_n s}\frac{\|\delta X_r\|^2_{\bP_n,2}}{\|\delta\|_2^2}\right)^{1/2}\\
&\gtrsim 1
\end{align*}
with probability $1-o(1)$ for a suitable choice of $q$ with $q>\bar{q}$. Since 
\begin{align*}
\frac{\lambda}{n}\lesssim n^{-1/2}\Phi^{-1}\left(1-\gamma/(2dp)\right)\lesssim n^{-1/2}\sqrt{\log(2dp/\gamma)}\lesssim n^{-1/2}\log^{\frac{1}{2}}(a_n)
\end{align*} and the penalty loading are uniformly bounded from above and away from zero we have
\begin{align*}
\max\limits_{r=1,\dots,d}\|(\hat{\beta}_r-\beta^{(1)}_r)X_r\|_{\bP_n,2}\lesssim_P L\sqrt{\frac{s\log(a_n)}{n}}
\end{align*}
by lemma $L.1$ from Belloni et al. (2018) \cite{belloni2018uniformly}. \\ \\
To establish assumption $L.1(b)$ for $m\ge 1$, we proceed by induction. Assume that the assumption holds for $\hat{\Psi}_{r,m-1}$ with some $\Delta_n=o(1)$, $l\gtrsim 1$ and $L\lesssim n^{1/q}\log^{\frac{1}{\rho}}(a_n)$. We have shown that the estimator based on $\hat{\Psi}_{r,m-1}$ obeys
\begin{align*}
\max\limits_{r=1,\dots,d}\|(\hat{\beta}_r-\beta^{(1)}_r)X_r\|_{\bP_n,2}\lesssim L\sqrt{\frac{s\log(a_n)}{n}}
\end{align*}
with probability $1-o(1)$. Observe that
\begin{align*}
\max\limits_{r=1,\dots,d}\|\beta^{(2)}_r X_r\|_{\bP_n,2}&\lesssim_P \sqrt{\frac{s\log(a_n)}{n}}
\end{align*}
as shown in (\ref{rateapprox}).
Using the triangle inequality we obtain with probability $1-o(1)$
\begin{align*}
\max\limits_{r=1,\dots,d}\|(\hat{\beta}_r-\beta_r)X_r\|_{\bP_n,2}&\le  \max\limits_{r=1,\dots,d}\|(\hat{\beta}_r-\beta^{(1)}_r)X_r\|_{\bP_n,2} + \max\limits_{r=1,\dots,d}\|\beta^{(2)}_r X_r\|_{\bP_n,2}\\
&\lesssim L\sqrt{\frac{s\log(a_n)}{n}}. 
\end{align*}
This implies 
\begin{align*}
|\hat{l}_{r,j,m}-l_{r,j}|&=\left|\E_n\left[\left(\left(Y_r-\hat{\beta}_r X_r\right) X_{r,j}\right)^2\right]^{1/2}-\E_n\left[\left(\left(Y_r-\beta_r X_r\right) X_{r,j}\right)^2\right]^{1/2}\right|\\
&\le \left|\E_n\left[\left(\left((\hat{\beta}_r-\beta_r) X_r\right) X_{r,j}\right)^2\right]^{1/2}\right|\\
&\lesssim \|(\hat{\beta}_r-\beta_r)X_r\|_{\bP_n,2} \max\limits_{1\le i \le n} \max\limits_{r=1,\dots,d}\|X_r^{(i)}\|_{\infty}\\
&\lesssim_P L\sqrt{\frac{s\log(a_n)}{n}} n^{1/q}\log^{\frac{1}{\rho}}(a_n)\\
&\lesssim \sqrt{n^{4/q}\frac{s\log^{1+\frac{4}{\rho}}(a_n)}{n}}=o(1)
\end{align*}
uniformly over $r=1,\dots,d$ and $j=1,\dots,p$.
Therefore assumption $L.1(b)$ holds for $\hat{\Psi}_{r,m}$ for some $\Delta_n=o(1)$, $l\gtrsim 1$ and $L\lesssim 1$.\\
Consequently, we have
\begin{align*}
\max\limits_{r=1,\dots,d}\|(\hat{\beta}_r-\beta^{(1)}_r)X_r\|_{\bP_n,2}\lesssim \sqrt{\frac{s\log(a_n)}{n}}.
\end{align*}
and
\begin{align*}
\max\limits_{r=1,\dots,d}\|\hat{\beta}_r-\beta^{(1)}_r\|_{1}\lesssim \sqrt{\frac{s^2\log(a_n)}{n}}
\end{align*}
with probability $1-o(1)$ due to lemma $L.1$ from Belloni et al. (2018) \cite{belloni2018uniformly}.\\ 
Observe that with probability $1-o(1)$ uniformly over all $r=1,\dots,d$ we have
\begin{align*}
&\left|\left(\E_n\left[\partial_{\beta}M_r(Y_r,X_r,\hat{\beta}_r)-\partial_{\beta}M_r(Y_r,X_r,\beta^{(1)}_r)\right]\right)^T\delta\right|\\
=&\left|\left(\E_n\left[(\hat{\beta}_r-\beta_r^{(1)})X_r X_r^T\right]\right)^T\delta\right|\\
\le&\|(\hat{\beta}_r-\beta^{(1)}_r)X_r\|_{\bP_n,2}\|\delta X_r\|_{\bP_n,2}\le L_n\|\delta X_r\|_{\bP_n,2}
\end{align*}
where $L_n\lesssim (s\log(a_n)/n)^{1/2}$. Since the maximal sparse eigenvalues
$$ \phi_{max}(l_ns,r):=\max\limits_{\|\delta\|_0\le l_n s}\frac{\|\delta X_r\|^2_{\bP_n,2}}{\|\delta\|_2^2}$$
are uniformly bounded from above, lemma $L.2$ from Belloni et al. (2018) \cite{belloni2018uniformly} directly implies
\begin{align*}
\max\limits_{r=1,\dots,d}\|\hat{\beta}_r\|_0 \lesssim s
\end{align*}
with probability $1-o(1)$. Combining this result with the uniform restrictions on the sparse eigenvalues from above we directly obtain
\begin{align*}
\max\limits_{r=1,\dots,d}\|\hat{\beta}_r-\beta_r^{(1)}\|_2 \lesssim \max\limits_{r=1,\dots,d}\|(\hat{\beta}_r-\beta^{(1)}_r)X_r\|_{\bP_n,2}\lesssim \sqrt{\frac{s\log(a_n)}{n}} 
\end{align*}
with probability $1-o(1)$.\\
We now proceed by using lemma $L.3$ from Belloni et al. (2018) \cite{belloni2018uniformly}. We obtain uniformly over all $r=1,\dots,d$
\begin{align*}
\E_n[M_r(Y_r,X_r,\tilde{\beta}_r)]-\E_n[M_r(Y_r,X_r,\beta_r)]&\le \frac{\lambda L}{n}\|\hat{\beta}_r-\beta_r\|_1\max\limits_{r=1,\dots,d}\|\hat{\Psi}^{(0)}_r\|_{\infty}\\
&\lesssim \frac{\lambda}{n}\|\hat{\beta}_r-\beta_r\|_1\\
&\lesssim \frac{s\log(a_n)}{n}
\end{align*}
with probability $1-o(1)$, where we used $L\lesssim 1$ and $\max\limits_{r=1,\dots,d}\|\hat{\Psi}^{(0)}_r\|_{\infty}\lesssim 1$. Since
\begin{align*}
\max\limits_{r=1,\dots,d}\|\E_n[S_r]\|_{\infty}&\le \max\limits_{r=1,\dots,d}\|\hat{\Psi}^{(0)}_r\|_{\infty}\|\big(\hat{\Psi}^{(0)}_r\big)^{-1}\E_n[S_r]\|_{\infty}\lesssim \frac{\lambda}{n}\lesssim n^{-1/2}\log^{\frac{1}{2}}(a_n) 
\end{align*}
with probability $1-o(1)$, we obtain
\begin{align*}
\max\limits_{r=1,\dots,d}\|(\tilde{\beta}_r-\beta^{(1)}_r)X_r\|_{\bP_n,2}\lesssim \sqrt{\frac{s\log(a_n)}{n}}
\end{align*} 
with probability $1-o(1)$, where we used 
$$\max\limits_{r=1,\dots,d}\|\hat{\beta}_r\|_0 \lesssim s ,\ C_n\lesssim (s\log(a_n)/n)^{1/2}$$
and that the minimum sparse eigenvalues are uniformly bounded away from zero. With the same argument as above we directly obtain 
\begin{align*}
\max\limits_{r=1,\dots,d}\|\tilde{\beta}_r-\beta_r^{(1)}\|_2 \lesssim \max\limits_{r=1,\dots,d}\|(\tilde{\beta}_r-\beta^{(1)}_r)X_r\|_{\bP_n,2}\lesssim \sqrt{\frac{s\log(a_n)}{n}}
\end{align*}
This finally completes the proof.
\end{proof}

\begin{proof}[Proof of Lemma \ref{condwl}]\ \\
See the proof of lemma L.4 from Belloni et al. (2018) \cite{belloni2018uniformly}. Since the regressors are standardized for all $j=1,\dots,p$ and independent from the error terms for all $r=1,\dots,d$, observe that
$$\frac{\E[X_{r,j}^2\varepsilon_r^2]}{\E[\varepsilon_r^2]}=\frac{\E[X_{r,j}^2]\E[\varepsilon_r^2]}{\E[\varepsilon_r^2]}=\E[X_{r,j}^2]=1.$$
We have due to \ref{WL}(iii)
\begin{align*}
&P\left(\max_{r=1,\dots,d}\max_{j=1,\dots,p}\frac{\E_n[X_{r,j}^2\varepsilon_r^2]}{\E_n[\varepsilon_r^2]}>1+\varphi_n\right)\\
\le &P\left(\max_{r=1,\dots,d}\max_{j=1,\dots,p}\frac{\E[X_{r,j}^2\varepsilon_r^2]+\tilde{\varphi}_n}{\E[\varepsilon_r^2]-\bar{\varphi}_n}>1+\varphi_n\right)+\Delta_n\\
\le &P\left(\max_{r=1,\dots,d}\left|\frac{\E[\varepsilon_r^2]+\tilde{\varphi}_n}{\E[\varepsilon_r^2]-\bar{\varphi}_n}-1\right|>\varphi_n\right)+\Delta_n\\
= &P\left(\max_{r=1,\dots,d}\left|\frac{\E[\varepsilon_r^2]+\tilde{\varphi}_n}{\E[\varepsilon_r^2]-\bar{\varphi}_n}-\frac{\E[\varepsilon_r^2]}{\E[\varepsilon_r^2]}\right|>\varphi_n\right)+\Delta_n\\
=&P\left(\max_{r=1,\dots,d}\left|\frac{\left(\E[\varepsilon_r^2]+\tilde{\varphi}_n\right)\E[\varepsilon_r^2]-\E[\varepsilon_r^2]\left(\E[\varepsilon_r^2]-\bar{\varphi}_n\right)}{\left(\E[\varepsilon_r^2]-\bar{\varphi}_n\right)\E[\varepsilon_r^2]}\right|>\varphi_n\right)+\Delta_n\\
=&\underbrace{P\left(\left|\frac{\left(\left(1+\tilde{\varphi}'_n\right)-\left(1-\bar{\varphi}'_n\right)\right)}{\left(1-\bar{\varphi}'_n\right)}\right|>\varphi_n\right)}_{=0}+\Delta_n,
\end{align*}
for an suitable choice of $\varphi_n =o(1)$, where $\bar{\varphi}'_n\ge \underline{C} \bar{\varphi}_n$ and $\tilde{\varphi}'_n\le \overline{C}\tilde{\varphi}_n$ due to \ref{WL}(ii) .\\
Next, for each $j=1,\dots,p$ and $r=1,\dots,d$, we apply lemma O.1 from Belloni et al. (2018) \cite{belloni2018uniformly} with $\mu=1$ and $\ell_n=c''\varphi_n^{-1}$, where $c''$ is a small constant that can be chosen to depend only on $\underline{C}$ and $\overline{C}$. Then conditions \ref{WL}(i) and \ref{WL}(ii) imply
\begin{align*}
0\le \Phi^{-1}\left(1-\frac{\gamma}{2pd}\right)\le \frac{n^{1/6}M_n(j,r)}{\ell_n}-1
\end{align*}
for $M_n(j,r)= \E[X_{r,j}^2\varepsilon_r^2]^{1/2}/\E[|X_{r,j}\varepsilon_r|^3]^{1/3}$ for each $r=1,\dots,d$ and $j=1,\dots,p$.\\
Therefore, we have
\begin{align*}
&P\left(c\max_{r=1,\dots,d}\|S_r\|_{\infty}>c'n^{-1/2}\Phi^{-1}\left(1-\frac{\gamma}{2pd}\right)\right)\\
=&P\left(c\max_{r=1,\dots,d}\max_{j=1,\dots,p}\frac{|\E_n[X_{r,j}\varepsilon_r]|}{\sqrt{\E_n[\varepsilon_r^2]}}>c'n^{-1/2}\Phi^{-1}\left(1-\frac{\gamma}{2pd}\right)\right)\\
\le&\sum\limits_{r=1}^{d}\sum\limits_{j=1}^p P\left(c\frac{|n^{1/2}\E_n[X_{r,j}\varepsilon_r]|}{\sqrt{\E_n[\varepsilon_r^2]}}>c'\Phi^{-1}\left(1-\frac{\gamma}{2pd}\right)\right)\\
=&\sum\limits_{r=1}^{d}\sum\limits_{j=1}^p P\left(c\frac{|n^{1/2}\E_n[X_{r,j}\varepsilon_r]|}{\sqrt{\E_n[X_{r,j}^2\varepsilon_r^2]}}\sqrt{\frac{\E_n[X_{r,j}^2\varepsilon_r^2]}{\E_n[\varepsilon_r^2]}}>c'\Phi^{-1}\left(1-\frac{\gamma}{2pd}\right)\right)\\
\le&\sum\limits_{r=1}^{d}\sum\limits_{j=1}^p P\left(\frac{|n^{1/2}\E_n[X_{r,j}\varepsilon_r]|}{\sqrt{\E_n[X_{r,j}^2\varepsilon_r^2]}}c\sqrt{1+\varphi_n}>c'\Phi^{-1}\left(1-\frac{\gamma}{2pd}\right)\right)+\Delta_n\\
\le& 2pd \frac{\gamma}{2pd}\left(1+O(\varphi_n^{1/3})\right)+\Delta_n\\
\le& \gamma +o(\gamma)+\Delta_n
\end{align*}
 for a sufficiently large $n$ (implying $c\sqrt{1+\varphi_n}\le c'$).
\end{proof}
\begin{proof}[Proof of Theorem \ref{uniform1}]\ \\
The proof is derived from the proof of lemma L.1. from Belloni et al. (2018) \cite{belloni2018uniformly}. At first we show that condition \ref{WL} is fulfilled. Conditions \ref{WL} (i), \ref{WL} (ii) and the first part of condition \ref{WL} (iii) have been verified in the proof of Theorem \ref{uniformlasso}. Hence, we need to show
\begin{align*}
\max_{r=1,\dots,d} |\E_n[\varepsilon_r^2]-\E[\varepsilon_r^2]|\le \bar{\varphi}_n
\end{align*}
with probability converging to one.\\
Let $\mathcal{W}=(\mathcal{Y},\mathcal{X})$ with $Y=(Y_1,\dots,Y_d)\in\mathcal{Y}$ and $X=(X_1,\dots,X_d)\in\mathcal{X}$.
Define $\mathcal{F}:=\{f_r|r=1,\dots,d\}$ with 
\begin{align*}
f_r:&\mathcal{W}=(\mathcal{Y},\mathcal{X})\to \mathbb{R}\\
&W=(Y,X)\mapsto (Y_r-\beta_r X_r)^2=\varepsilon^2_r.
\end{align*}
For a constant C that does depend on $q$ but not on $n$, observe that  
$$F:=\|\sup_{f\in \mathcal{F}}|f|\|_{P,q}=\|\max_{r=1,\dots,d}\varepsilon_r^2\|_{P,q}= \left(\E\left[\max_{r=1,\dots,d} \varepsilon_r^{2q}\right]^{1/2q}\right)^2\le C\log(d)^{\frac{2}{\rho}}$$
where we used the same argument as in the beginning of the proof of Theorem \ref{uniformlasso}.\\
Due to Assumption \ref{tails} the second moments of the error terms are uniformly bounded and hence we can choose a constant $C$ such that
$$\max_{r=1,\dots,d}\|\varepsilon_r\|_{P,2}^2\le C\le \|\max_{r=1,\dots,d}\varepsilon_r^2\|_{P,q}$$
and since $|\mathcal{F}|=d$ we have
$$\log\sup_{Q} N(\varepsilon\|F\|_{Q,2},\mathcal{F},\|\cdot\|_{Q,2})\le \log(d).$$
Therefore we are able to use lemma O.2 from Belloni et al. (2018) \cite{belloni2018uniformly}, which implies that with probability $1-o(1)$
\begin{align*}
\max_{r=1,\dots,d} |\E_n[\varepsilon_r^2]-\E[\varepsilon_r^2]|&=n^{-1/2}\sup_{f\in\mathcal{F}}|\mathbb{G}_n(f)|\\
&\lesssim \left(\sqrt{\frac{\log(d)}{n}}+ \frac{\log^{1+\frac{2}{\rho}}(d)}{n^{1-1/q}}\right)\le \bar{\varphi}_n.
\end{align*}
Due to the definition of $\hat{\beta}_r$ we have
$$\hat{Q}_r^{1/2}(\hat{\beta}_r)+\frac{\lambda}{n}\|\hat{\beta}_r\|_1\le \hat{Q}_r^{1/2}(\beta_r^{(1)})+\frac{\lambda}{n}\|\beta_r^{(1)}\|_1 $$
implying
\begin{align}\label{eq1} 
\hat{Q}_r^{1/2}(\hat{\beta}_r)- \hat{Q}_r^{1/2}(\beta_r^{(1)})\le \frac{\lambda}{n}\left(\|\delta_{r,T_r}\|_1-\|\delta_{r,T_r^c}\|_1\right)
\end{align}
with  $\delta_r:=\hat{\beta}_r-\beta_r^{(1)}$. Due to the convexity of $\beta\mapsto\hat{Q}_r^{1/2}(\beta)$ we have with probability $1-o(1)$:
\begin{align*}
\hat{Q}_r^{1/2}(\hat{\beta}_r)- \hat{Q}_r^{1/2}(\beta_r^{(1)})\ge \delta_r \hat{S}_r.
\end{align*}
For a sequence $C_n\lesssim\sqrt{\frac{s\log(a_n)}{n}}$ independent from $r$, it holds
\begin{align*}
|\delta_r \hat{S}_r|&\le |\delta_r S_r|+|\delta_r (\hat{S}_r-S_r)|\\
&\lesssim_P \|\delta_r\|_1\frac{\lambda}{nc} +|\delta_r (\hat{S}_r-S_r)|\\
&\lesssim_P\|\delta_r\|_1\frac{\lambda}{nc} +C_n\|\delta_r X_r\|_{\mathbb{P}_n,2}.
\end{align*}
To obtain the last inequality observe that
\begin{align*}
\E_n[(\varepsilon_r+\beta^{(2)}_rX_r)^2]&=\E_n[\varepsilon_r^2]+2\E_n[\varepsilon_r\beta^{(2)}_rX_r]+\underbrace{\E_n[(\beta^{(2)}_rX_r)^2]}_{\ge 0}\\
&\gtrsim \min_{r=1,\dots,d}\E[\varepsilon_r^2] + o_P(1)\\
&\gtrsim c + o_P(1)
\end{align*}
is uniformly bounded away from zero, since with probability $1-o(1)$
\begin{align*}
\min_{r=1,\dots,d}\E_n[\varepsilon_r\beta^{(2)}_rX_r]&\ge -\max_{r=1,\dots,d}|\E_n[\varepsilon_r\beta^{(2)}_rX_r]|\\
&\ge -\max_{r=1,\dots,d}\sqrt{\E_n[\varepsilon_r^2]\E_n[(\beta^{(2)}_rX_r)^2]}\\
&\gtrsim -\sqrt{\left(\max_{r=1,\dots,d}\E[\varepsilon_r^2]+\bar{\varphi}_n\right)\left(\max_{r=1,\dots,d}\E[(\beta^{(2)}_rX_r)^2]+\frac{s\log(a_n)}{n}\right)}\\
&\gtrsim -\sqrt{\frac{s\log(a_n)}{n}}
\end{align*}
uniformly converges towards zero where we used that
$$\max_{r=1,\dots,d} |\E_n[(\beta^{(2)}_rX_r)^2]-\E[(\beta^{(2)}_rX_r)^2]|\lesssim_P \frac{s\log(a_n)}{n}$$
as shown in proof of Theorem \ref{uniformlasso}.\\
This implies that
{\allowdisplaybreaks
\begin{align*}
|\delta_r (\hat{S}_r-S_r)| &= \left|\delta_r\left(\frac{\E_n[X_r(\varepsilon_r+\beta^{(2)}_rX_r)]}{\sqrt{\E_n[(\varepsilon_r+\beta^{(2)}_rX_r)^2]}}-\frac{\E_n[X_r\varepsilon_r]}{\sqrt{\E_n[\varepsilon_r^2]}}\right)\right|\\
&=\left|\delta_r\frac{\E_n[X_r(\varepsilon_r+\beta^{(2)}_rX_r)]\sqrt{\E_n[\varepsilon_r^2]}-\E_n[X_r\varepsilon_r]\sqrt{\E_n[(\varepsilon_r+\beta^{(2)}_rX_r)^2]}}{\sqrt{\E_n[(\varepsilon_r+\beta^{(2)}_rX_r)^2]\E_n[\varepsilon_r^2]}}\right|\\
&\lesssim_P \bigg|\delta_r\bigg(\E_n[X_r(\beta^{(2)}_rX_r)]\sqrt{\E_n[\varepsilon_r^2]}\\
&\quad+\E_n[X_r\varepsilon_r]\left(\sqrt{\E_n[\varepsilon_r^2]}-\sqrt{\E_n[(\varepsilon_r+\beta^{(2)}_rX_r)^2]}\right)\bigg)\bigg|\\
&\le \left|\E_n[(\delta_rX_r)(\beta^{(2)}_rX_r)]\sqrt{\E_n[\varepsilon_r^2]}\right|\\
&\quad+\left|\E_n[(\delta_rX_r)\varepsilon_r]\right|\underbrace{\Big|\Big(\sqrt{\E_n[\varepsilon_r^2]}-\sqrt{\E_n[(\varepsilon_r+\beta^{(2)}_rX_r)^2]}\Big)\Big|}_{\le\sqrt{\E_n[(\beta^{(2)}_rX_r)^2]}}\\
&\lesssim \sqrt{\E_n[(\delta_rX_r)^2]\E_n[(\beta^{(2)}_rX_r)^2]\E_n[\varepsilon_r^2]}\\
&\lesssim_P C_n \|\delta_r X_r\|_{\mathbb{P}_n,2}
\end{align*}
}\noindent
with an analogous argument as above. Hence, we have with probability $1-o(1)$
\begin{align}\label{eq2}
\hat{Q}_r^{1/2}(\hat{\beta}_r)- \hat{Q}_r^{1/2}(\beta_r^{(1)})\ge \delta_r \hat{S}_r\gtrsim -\|\delta_r\|_1\frac{\lambda}{nc}-C_n\|\delta_r X_r\|_{\mathbb{P}_n,2}.
\end{align}
Combining the inequalities (\ref{eq1}) and (\ref{eq2}) we obtain
\begin{align}
&-\|\delta_r\|_1\frac{\lambda}{nc}-C_n\|\delta_r X_r\|_{\mathbb{P}_n,2}\lesssim_P \frac{\lambda}{n}\left(\|\delta_{r,T_r}\|_1-\|\delta_{r,T_r^c}\|_1\right)\notag\\
\iff &\|\delta_{r,T_r^c}\|_1\lesssim_P \underbrace{\frac{c+1}{c-1}}_{:= \tilde{c}}\|\delta_{r,T_r}\|_1+\frac{n}{\lambda}\frac{c}{c-1}C_n\|\delta_r X_r\|_{\mathbb{P}_n,2}.\label{1norm_support}
\end{align}
Further we have
\begin{align*}
\hat{Q}_r(\hat{\beta}_r)- \hat{Q}_r(\beta_r^{(1)})=\|\delta_r X_r\|_{\mathbb{P}_n,2}^2-2\E_n[(Y_r-\beta_r^{(1)}X_r)\delta_r X_r]
\end{align*}
with 
\begin{align*}
\E_n[(Y_r-\beta_r^{(1)}X_r)\delta_r X_r]&=\E_n[\varepsilon_r\delta_r X_r]+\E_n[(\beta_r^{(2)}X_r)\delta_r X_r]\\
&\lesssim_P Q_r^{1/2}(\beta_r^{(1)})||S_r||_\infty||\delta_r||_1+C_n\|\delta_r X_r\|_{\mathbb{P}_n,2}
\end{align*}
by H\"older inequality. Due to Lemma $P.1$ in \cite{belloni2018uniformly} with $K\lesssim n^{1/\bar{q}}\log^{\frac{1}{\rho}}(a_n)$, $k\lesssim s$ for a suitable $\bar{q}>\tilde{q}$ and
\begin{align*}
\delta_n&\lesssim K\sqrt{s}n^{-1/2}\log(s)\log^{1/2}(a_n)\log^{1/2}(n)\\
&\lesssim \sqrt{n^{\frac{1}{\tilde{q}}}\frac{s\log^{1+\frac{2}{\rho}(a_n)}}{n}}=o(1)
\end{align*}
by growth condition \ref{growthc}, it holds
$$ c\le\phi_{min}(k,r)\le\phi_{max}(k,r)\le C$$
with probability $1-o(1)$ uniformly over $r=1,\dots,d$. Hence, the restricted eigenvalue
$$ \kappa_{2\tilde{c}}=\min_{r=1,\dots,d}\inf_{\delta\in\Delta_{2\tilde{c},r}}\frac{\|\delta X_r\|_{\mathbb{P}_n,2}}{\|\delta\|_2}$$
is bounded away from zero with probability $1-o(1)$ where $$\Delta_{2\tilde{c},r}=\{\delta:||\delta_{T_r^c}||_1\le2\tilde{c}||\delta_{T_r}||_1\}.$$
If $\delta_r\in\Delta_{2\tilde{c},r}$, then
\begin{align*}
\|\delta_r X_r\|_{\mathbb{P}_n,2}^2&=2\E_n[(Y_r-\beta_r^{(1)}X_r)\delta_rX_r]+[\hat{Q}_r^{1/2}(\hat{\beta}_r)+\hat{Q}_r^{1/2}(\beta_r^{(1)})][\hat{Q}_r^{1/2}(\hat{\beta}_r)- \hat{Q}_r^{1/2}(\beta_r^{(1)})]\\
&\lesssim_P 2Q_r^{1/2}(\beta_r^{(1)})||S_r||_\infty||\delta_r||_1+2C_n\|\delta_r X_r\|_{\mathbb{P}_n,2}\\
&\quad+[\hat{Q}_r^{1/2}(\hat{\beta}_r)+\hat{Q}_r^{1/2}(\beta_r^{(1)})]\frac{\lambda}{n}\left(\frac{\sqrt{s}||\delta_rX_r||_{\mathbb{P}_n,2}}{\kappa_{2\tilde{c}}}-||\delta_{r,T_r^c}||_1)\right).
\end{align*}
Using
\begin{align*}
\hat{Q}_r^{1/2}(\hat{\beta}_r)\le\hat{Q}_r^{1/2}(\beta_r^{(1)})+\frac{\lambda}{n}\frac{\sqrt{s}||\delta_rX_r||_{\mathbb{P}_n,2}}{\kappa_{2\tilde{c}}}
\end{align*}
we conclude
\begin{align*}
\|\delta_r X_r\|_{\mathbb{P}_n,2}^2&\lesssim_P 2Q_r^{1/2}(\beta_r^{(1)})||S_r||_\infty||\delta_r||_1\\
&\quad+\left[2\hat{Q}_r^{1/2}(\beta_r^{(1)})+\frac{\lambda}{n}\frac{\sqrt{s}||\delta_r||_{\mathbb{P}_n,2}}{\kappa_{2\tilde{c}}}\right]\frac{\lambda}{n}\left(\frac{\sqrt{s}||\delta_r||_{\mathbb{P}_n,2}}{\kappa_{2\tilde{c}}}-||\delta_{r,T_r^c}||_1)\right)  \\
&\quad+2C_n\|\delta_r X_r\|_{\mathbb{P}_n,2}\\
&\lesssim_P 2\frac{\lambda}{n}\left(Q_r^{1/2}(\beta_r^{(1)})||\delta_r||_1-\hat{Q}_r^{1/2}(\beta_r^{(1)})||\delta_{r,T_r^c}||_1\right)\\
&\quad+2\hat{Q}_r^{1/2}(\beta_r^{(1)})\frac{\lambda}{n}\frac{\sqrt{s}||\delta_rX_r||_{\mathbb{P}_n,2}}{\kappa_{2\tilde{c}}}+\left(\frac{\lambda}{n}\frac{\sqrt{s}||\delta_rX_r||_{\mathbb{P}_n,2}}{\kappa_{2\tilde{c}}}\right)^2+2C_n\|\delta_r X_r\|_{\mathbb{P}_n,2}
\end{align*}
with
\begin{align*}
&\left(Q_r^{1/2}(\beta_r^{(1)})||\delta_r||_1-\hat{Q}_r^{1/2}(\beta_r^{(1)})||\delta_{r,T_r^c}||_1\right)\\
&=\hat{Q}_r^{1/2}(\beta_r^{(1)})||\delta_{r,T_r}||_1+\left(Q_r^{1/2}(\beta_r^{(1)})-\hat{Q}_r^{1/2}(\beta_r^{(1)})\right)||\delta_r||_1\\
&\le\hat{Q}_r^{1/2}(\beta_r^{(1)})||\delta_{r,T_r}||_1+\|\beta_r^{(2)} X_r\|_{\mathbb{P}_n,2}||\delta_r||_1\\
&\lesssim_P\hat{Q}_r^{1/2}(\beta_r^{(1)})||\delta_{r,T_r}||_1+C_n3\tilde{c}||\delta_{r,T_r}||_1.
\end{align*}
With probability $1-o(1)$ we have
\begin{align*}
\|\delta_r X_r\|_{\mathbb{P}_n,2}^2&\lesssim 2\frac{\lambda}{n}||\delta_{r,T_r}||_1\left(\hat{Q}_r^{1/2}(\beta_r^{(1)})+C_n3\bar{c}\right)\\
&\quad+2\hat{Q}_r^{1/2}(\beta_r^{(1)})\frac{\lambda}{n}\frac{\sqrt{s}||\delta_rX_r||_{\mathbb{P}_n,2}}{\kappa_{2\tilde{c}}}+\left(\frac{\lambda}{n}\frac{\sqrt{s}||\delta_rX_r||_{\mathbb{P}_n,2}}{\kappa_{2\tilde{c}}}\right)^2+2C_n\|\delta_r X_r\|_{\mathbb{P}_n,2}\\
&\lesssim 2\frac{\lambda}{n}\frac{\sqrt{s}||\delta_rX_r||_{\mathbb{P}_n,2}}{\kappa_{2\tilde{c}}}\left(\hat{Q}_r^{1/2}(\beta_r^{(1)})+C_n3\bar{c}\right)\\
&\quad+2\hat{Q}_r^{1/2}(\beta_r^{(1)})\frac{\lambda}{n}\frac{\sqrt{s}||\delta_rX_r||_{\mathbb{P}_n,2}}{\kappa_{2\tilde{c}}}+\left(\frac{\lambda}{n}\frac{\sqrt{s}||\delta_rX_r||_{\mathbb{P}_n,2}}{\kappa_{2\tilde{c}}}\right)^2+2C_n\|\delta_r X_r\|_{\mathbb{P}_n,2}
\end{align*}
and therefore obtain
\begin{align*}
\left(1-\left(\frac{\lambda}{n}\frac{\sqrt{s}}{\kappa_{2\tilde{c}}}\right)^2\right)\|\delta_r X_r\|_{\mathbb{P}_n,2}^2&\lesssim_P \bigg(4\hat{Q}_r^{1/2}(\beta_r^{(1)})\frac{\lambda}{n}\frac{\sqrt{s}}{\kappa_{2\tilde{c}}}\\
&\quad+C_n\left(6\tilde{c}\frac{\lambda}{n}\frac{\sqrt{s}}{\kappa_{2\tilde{c}}}+2\right)\bigg)||\delta_rX_r||_{\mathbb{P}_n,2},
\end{align*}
which implies
\begin{align*}
\|\delta_r X_r\|_{\mathbb{P}_n,2}&\lesssim_P \frac{\lambda\sqrt{s}}{n}+C_n\lesssim \sqrt{\frac{s\log(a_n)}{n}}.
\end{align*}
Here we used that 
\begin{align*}
\hat{Q}_r^{1/2}(\beta_r^{(1)})=\E_n[(\varepsilon_r+\beta_r^{(2)}X_r)^2]^{1/2}\le\|\varepsilon_r\|_{\mathbb{P}_n,2}+\|\beta_r^{(2)}X_r\|_{\mathbb{P}_n,2}\lesssim_P C+\bar{\varphi}_n+C_n.
\end{align*}
If $\delta_r\notin\Delta_{2\tilde{c},r}$ (implying $||\delta_{r,T_r^c}||_1> 2\tilde{c}||\delta_{r,T_r}||_1$), (\ref{1norm_support}) directly implies 
$$ 2\tilde{c}||\delta_{r,T_r}||_1\lesssim_P \tilde{c}\|\delta_{r,T_r}\|_1+\frac{n}{\lambda}\frac{c}{c-1}C_n\|\delta_r X_r\|_{\mathbb{P}_n,2}$$
and therefore
$$||\delta_{r,T_r}||_1\lesssim_P \frac{n}{\lambda}\frac{c}{c-1}C_n\|\delta_r X_r\|_{\mathbb{P}_n,2}$$
due to $\tilde{c}\ge 1$. Additionally (\ref{1norm_support}) implies
$$\|\delta_{r,T_r^c}\|_1\lesssim_P \frac{1}{2}\|\delta_{r,T_r^c}\|_1+\frac{n}{\lambda}\frac{c}{c-1}C_n\|\delta_r X_r\|_{\mathbb{P}_n,2}$$
and therefore
$$\|\delta_{r,T_r^c}\|_1\lesssim_P\frac{2n}{\lambda}\frac{c}{c-1}C_n\|\delta_r X_r\|_{\mathbb{P}_n,2},$$
which, combined with the inequality above, implies
$$\|\delta_{r}\|_1\lesssim_P\frac{3n}{\lambda}\frac{c}{c-1}C_n\|\delta_r X_r\|_{\mathbb{P}_n,2}.$$
Using
\begin{align*}
\hat{Q}_r^{1/2}(\hat{\beta}_r)-\hat{Q}_r^{1/2}(\beta_r^{(1)})\le\frac{\lambda}{n}\left(\|\delta_{r,T_r}\|_1-\|\delta_{r,T_r^c}\|_1\right)\le \frac{\lambda}{n}\|\delta_{r}\|_1
\end{align*}
and following the same argument as above we obtain with probability $1-o(1)$:
\begin{align*}
\|\delta_r X_r\|_{\mathbb{P}_n,2}^2&=2\E_n[(Y_r-\beta_r^{(1)}X_r)\delta_rX_r]+[\hat{Q}_r^{1/2}(\hat{\beta}_r)+\hat{Q}_r^{1/2}(\beta_r^{(1)})][\hat{Q}_r^{1/2}(\hat{\beta}_r)- \hat{Q}_r^{1/2}(\beta_r^{(1)})]\\
&\lesssim 2Q_r^{1/2}(\beta_r^{(1)})||S_r||_\infty||\delta_r||_1+2C_n\|\delta_r X_r\|_{\mathbb{P}_n,2}\\
&\quad+\left(2\hat{Q}_r^{1/2}(\beta_r^{(1)})+\frac{\lambda}{n}\|\delta_{r}\|_1\right)\frac{\lambda}{n}\|\delta_{r}\|_1\\
&\lesssim \bigg(2\frac{1}{c}\underbrace{\left(Q_r^{1/2}(\beta_r^{(1)})-\hat{Q}_r^{1/2}(\beta_r^{(1)})\right)}_{\lesssim C_n}+2\left(\frac{1}{c}+1\right)\hat{Q}_r^{1/2}(\beta_r^{(1)})+\frac{\lambda}{n}\|\delta_{r}\|_1\bigg)\frac{\lambda}{n}\|\delta_{r}\|_1\\
&\quad+2C_n\|\delta_r X_r\|_{\mathbb{P}_n,2}\\
&\le 6\left(\frac{C_n}{c}+\left(\frac{1}{c}+1\right)\hat{Q}_r^{1/2}(\beta_r^{(1)})\right)\frac{c}{c-1}C_n\|\delta_r X_r\|_{\mathbb{P}_n,2}\\
&\quad+\left(3\frac{c}{c-1}C_n\|\delta_r X_r\|_{\mathbb{P}_n,2}\right)^2+2C_n\|\delta_r X_r\|_{\mathbb{P}_n,2}.\\
\end{align*}
Hence,
\begin{align*}
\left(1-\left(3\frac{c}{c-1}C_n\right)^2\right)\|\delta_r X_r\|_{\mathbb{P}_n,2}^2&\lesssim_P 6\left(\frac{C_n}{c}+\left(\frac{1}{c}+1\right)\hat{Q}_r^{1/2}(\beta_r^{(1)})\right)\frac{c}{c-1}C_n\|\delta_r X_r\|_{\mathbb{P}_n,2}\\
&\quad +2C_n\|\delta_r X_r\|_{\mathbb{P}_n,2}\\
\end{align*}
which implies
\begin{align*}
\|\delta_r X_r\|_{\mathbb{P}_n,2}&\lesssim_P C_n\lesssim \sqrt{\frac{s\log(a_n)}{n}}.
\end{align*}
To prove the second claim observe that
\begin{align*}
\|\delta_r\|_1&= 1_{\{\delta_r\in\Delta_{2\tilde{c},r}\}}\|\delta_r\|_1+1_{\{\delta_r\notin\Delta_{2\tilde{c},r}\}}\|\delta_r\|_1\\
&\le 1_{\{\delta_r\in\Delta_{2\tilde{c},r}\}}\left(1+2\tilde{c}\right)\|\delta_{r,T_r}\|_1+1_{\{\delta_r\notin\Delta_{2\tilde{c},r}\}}\|\delta_r\|_1\\
&\lesssim_P \left(\left(1+2\tilde{c}\right)\frac{\sqrt{s}}{\kappa_{2\tilde{c}}}+\frac{3n}{\lambda}\frac{c}{c-1}C_n\right)\|\delta_r X_r\|_{\mathbb{P}_n,2}\\
&\lesssim_P \sqrt{\frac{s^2\log(a_n)}{n}}
\end{align*}
uniformly over all $r=1,\dots,d$. Now, we proof that
\begin{align*}
\max\limits_{r=1,\dots,d}\|\hat{\beta}_r\|_0\lesssim s.
\end{align*}
This proof is derived from the proof of lemma L.2. from Belloni et al. (2018) \cite{belloni2018uniformly}. At first observe that
$$0<c\lesssim_P  \min_{r=1,\dots,d}\|\varepsilon_r+\beta^{(2)}_r X_r\|^2_{\mathbb{P}_n,2}\le \max_{r=1,\dots,d}\|\varepsilon_r+\beta^{(2)}_r X_r\|^2_{\mathbb{P}_n,2}\lesssim_P C<\infty$$
where the first inequality is shown above and the second follows with an analogous argument. Additionally we obtain
\begin{align*}
\max_{r=1,\dots,d}\left|\|Y_r-\hat{\beta}_r X_r\|^2_{\mathbb{P}_n,2}-\|\varepsilon_r+\beta^{(2)}_r X_r\|^2_{\mathbb{P}_n,2}\right|\lesssim_P C_n+C_n^2=o(1)
\end{align*}
due to 
\begin{align*}
\|Y_r-\hat{\beta}_r X_r\|^2_{\mathbb{P}_n,2}&=\|\varepsilon_r+\beta^{(2)}_r X_r\|^2_{\mathbb{P}_n,2}-2\E_n[(\varepsilon_r+\beta^{(2)}_r X_r)\delta_rX_r]+\underbrace{\|\delta_rX_r\|^2_{\mathbb{P}_n,2}}_{\lesssim_P C_n^2}
\end{align*}
with
\begin{align*}
|\E_n[(\varepsilon_r+\beta^{(2)}_r X_r)\delta_rX_r]|&\le \sqrt{\E_n[(\varepsilon_r+\beta^{(2)}_r X_r)^2]\E_n[(\delta_rX_r)^2]}\\
&\lesssim \left(C + o_P(1)\right)\|\delta_rX_r\|_{\mathbb{P}_n,2}\\
&\lesssim_P C_n
\end{align*}
uniformly over all $r=1,\dots,d$. This implies
\begin{align*}
&|\delta(\partial_{\beta}\hat{Q}^{1/2}_r(\beta)|_{\beta=\hat{\beta}_r}-\hat{S}_r)|\\
=&\bigg|\delta\left(\frac{\E_n[X_r(Y_r-\beta_r^{(1)} X_r)]}{\sqrt{\E_n[(Y_r-\beta_r^{(1)} X_r)^2]}}-\frac{\E_n[X_r(Y_r-\hat{\beta}_r X_r)]}{\sqrt{\E_n[(Y_r-\hat{\beta}_r X_r)^2]}}\right)\bigg|\\
=& \bigg|\delta\left(\frac{\E_n[X_r(Y_r-\beta_r^{(1)} X_r)]\|Y_r-\hat{\beta}_r X_r\|_{\mathbb{P}_n,2}-\|\varepsilon_r+\beta^{(2)}_r X_r\|_{\mathbb{P}_n,2}\E_n[X_r(Y_r-\hat{\beta}_r X_r)]}{\|\varepsilon_r+\beta^{(2)}_r X_r\|_{\mathbb{P}_n,2}\|Y_r-\hat{\beta}_r X_r\|_{\mathbb{P}_n,2}}\right)\bigg|\\
\lesssim_P & \bigg|\delta \left(\E_n[X_r(Y_r-\beta_r^{(1)} X_r)]-\E_n[X_r(Y_r-\hat{\beta}_r X_r)]\right)\bigg|\\
\le & \|\delta_r X_r\|_{\mathbb{P}_n,2}\|\delta X_r\|_{\mathbb{P}_n,2}\lesssim_P C_n\|\delta X_r\|_{\mathbb{P}_n,2}.
\end{align*}
By the definition of $\hat{\beta_r}$, there exists a subgradient $\partial_{\beta}\hat{Q}^{1/2}_r(\beta)|_{\beta=\hat{\beta}_r}$ of $\hat{Q}^{1/2}_r(\hat{\beta}_r)$ such that for every $j$ with $|\hat{\beta}_{r,j}|>0$
$$|(\partial_{\beta}\hat{Q}^{1/2}_r(\beta)|_{\beta=\hat{\beta}_r})_j|=\frac{\lambda}{n}.$$
Let $\hat{T}_r:=\text{supp}(\hat{\beta}_r)$ and $|\hat{T}_r|:=\hat{s}_r$. We obtain
\begin{align*}
\frac{\lambda}{n}\sqrt{\hat{s}_r}&=\|(\partial_{\beta}\hat{Q}^{1/2}_r(\beta)|_{\beta=\hat{\beta}_r})_{\hat{T}_r}\|_2\\
&\le\|{S_r}_{\hat{T}_r}\|_2+\|(\hat{S}_r-S_r)_{\hat{T}_r}\|_2+\|(\partial_{\beta}\hat{Q}^{1/2}_r(\beta)|_{\beta=\hat{\beta}_r}-\hat{S}_r)_{\hat{T}_r}\|_2\\
&\lesssim_P \sqrt{\hat{s}_r}\|S_r\|_\infty\\
&\quad + C_n\sup_{\|\delta\|_2=1,\|\delta\|_0\le\hat{s}_r}\|\delta X_r\|_{\mathbb{P}_n,2}\\
&\quad + \sup_{\|\delta\|_2=1,\|\delta\|_0\le\hat{s}_r}|\delta(\partial_{\beta}\hat{Q}^{1/2}_r(\beta)|_{\beta=\hat{\beta}_r}-\hat{S}_r)|\\
&\lesssim_P  \sqrt{\hat{s}_r}\frac{\lambda}{nc}+2C_n\sup_{\|\delta\|_2=1,\|\delta\|_0\le\hat{s}_r}\|\delta X_r\|_{\mathbb{P}_n,2},
\end{align*}
where we used
\begin{align*}
\|(\hat{S}_r-S_r)_{\hat{T}_r}\|_2\le \sup_{\|\delta\|_2=1,\|\delta\|_0\le\hat{s}_r}|\delta(\hat{S}_r-S_r)|\lesssim_P C_n \sup_{\|\delta\|_2=1,\|\delta\|_0\le\hat{s}_r}\|\delta X_r\|_{\mathbb{P}_n,2}.
\end{align*}
Hence with probability $1-o(1)$,
\begin{align}\label{hats}
\hat{s}_r\le& \left(\frac{2CnC_n}{\lambda(1-1/c)}\right)^2\sup_{\|\delta\|_2=1,\|\delta\|_0\le\hat{s}_r}\|\delta X_r\|_{\mathbb{P}_n,2}^2\notag \\
\le&\bigg(\underbrace{\frac{2CnC_n}{\lambda(1-1/c)}}_{:=L}\bigg)^2\phi_{max}(\hat{s}_r,r)\lesssim s\phi_{max}(\hat{s}_r,r)
\end{align}
where
$$ \phi_{max}(\hat{s}_r,r):=\max\limits_{\|\delta\|_0\le \hat{s}_r}\frac{\|\delta X_r\|_{\bP_n,2}^2}{\|\delta\|_2^2}.$$
We can find a suitable $C$ such that $M=Cs\in\mathcal{M}_r$ with $$\mathcal{M}_r:=\{m\in\N:m>2\phi_{max}(m,r)L^2\}.$$
Suppose that $\hat{s}_r>M$. By the sublinearity of the maximum sparse eigenvalue (Lemma 3 in \cite{belloni:2013}), for any integer $k\ge 0$ and constant $l\ge 0$, we have 
$$\phi_{max}(lk,r)\le \ceil{l}\phi_{max}(k,r)$$
where $\ceil{l}$ denotes the ceiling of $l$. Since $\ceil{k}\le 2k$ for any $k\ge 1$,
\begin{align*}
\hat{s}_r\le& L^2\phi_{max}(\hat{s}_r,r)=L^2\phi_{max}(M\hat{s}_r/M,r)\\
\le& \ceil[\bigg]{\frac{\hat{s}_r}{M}}L^2\phi_{max}(M,r)\le\frac{2\hat{s}_r}{M}L^2\phi_{max}(M,r)
\end{align*}
that violates the condition that $M\in \mathcal{M}_r$. Therefore, we have $\hat{s}_r\le M$. Applying \ref{hats}, we obtain
\begin{align*}
\max\limits_{r=1,\dots,d}\hat{s}_r\le\max\limits_{r=1,\dots,d}\phi_{max}(M,r)s\lesssim s
\end{align*}
with probability $1-o(1)$ and the stated claim follows:
\begin{align*}
\max\limits_{r=1,\dots,d}\|\hat{\beta}_r\|_0\lesssim s.
\end{align*}
Since the maximal sparse eigenvalues are uniformly bounded from above, we conclude
\begin{align*}
\max\limits_{r=1,\dots,d}\|\hat{\beta}_r-\beta_r^{(1)}\|_2 \lesssim \max\limits_{r=1,\dots,d}\|(\hat{\beta}_r-\beta^{(1)}_r)X_r\|_{\bP_n,2}\lesssim C_n
\end{align*}
with probability at least $1-o(1)$.
\end{proof}

\footnotesize

\pagebreak
\bibliographystyle{imsart-number}
\bibliography{Literatur_GGM}

\end{document}